\documentclass[11pt]{article}
\usepackage{fullpage}
\usepackage{url}
\usepackage{xspace}

\usepackage{graphics}
\usepackage[dvips]{epsfig}

\usepackage{amsmath}
\usepackage{amssymb}
\usepackage{amsfonts}
\usepackage{graphicx}
 \usepackage{algorithm}
 \usepackage{algpseudocode}
 \usepackage{multirow}
 \usepackage{subcaption}
 \usepackage{diagbox}
\usepackage{enumitem}

\newtheorem{theorem}{Theorem}[section]
\newtheorem{lemma}{Lemma}[section]

\newtheorem{claim}{Claim}[section]
\newtheorem{proposition}{Proposition}[section]

\newtheorem{observation}{Observation}

\newcommand{\qed}{\hfill $\Box$ \bigbreak}
\newenvironment{proof}{\noindent {\bf Proof.}}{\qed}

\newcommand{\cA}{{\cal A}}

\newcommand{\remove}[1]{}

\def\idtt#1{\ensuremath{\mathtt{#1}}}
\def\len{\idtt{len}}

\def\leftFast{\idtt{left\_fast}}
\def\rightSlow{\idtt{right\_slow}}
\def\rightFast{\idtt{right\_fast}}

\def\negFast{\idtt{(-1)-fast}}
\def\plusFast{\idtt{(+1)-fast}}
\def\negSlow{\idtt{(-1)-slow}}
\def\plusSlow{\idtt{(+1)-slow}}

\def\idle{\idtt{idle}}
\def\senti{\idtt{sentinel}}
\def\hiber{\idtt{hibernate}}

\def\negColl{\idtt{(-1)-collect}}
\def\plusColl{\idtt{(+1)-collect}}

\def\maxDist{\idtt{max\_dist}}

\begin{document}

\baselineskip  0.19in 
\parskip     0.05in 
\parindent   0.3in 

\title{{\bf Gathering Teams of Bounded Memory Agents \\ on a Line
\footnote{A preliminary version of this paper appeared at the 28th International Conference on Principles of Distributed Systems (OPODIS 2024).}
 }}
\date{}
\newcommand{\inst}[1]{$^{#1}$}

\author{
Younan Gao\inst{1}$^,$\footnote{The work was done while Younan Gao was a postdoctoral researcher at the Universit\'e du Qu\'{e}bec en Outaouais. Younan Gao was partially supported by the Research Chair in Distributed Computing at the Université du Québec en Outaouais and MUR 2022YRB97K, PINC, Pangenome Informatics: from Theory to Applications.},
Andrzej Pelc\inst{2}$^,$\footnote{Partially supported by NSERC discovery grant RGPIN-2024-03767 and by the Research Chair in Distributed Computing at the Universit\'e du Qu\'{e}bec en Outaouais.}\\
\inst{1} University of Milano-Bicocca, Milan, Italy.\\
\inst{2} Universit\'{e} du Qu\'{e}bec en Outaouais, Gatineau, Canada.\\
E-mails: \url{younan.gao@unimib.it}, \url{ pelc@uqo.ca}\\
}

\date{ }
\maketitle

\begin{abstract}
 Several mobile agents, modelled as deterministic automata, navigate in an infinite line in synchronous rounds.
 All agents start in the same round. In each round, an agent can move to one of the two neighboring nodes, or stay idle.  Agents have distinct labels which are integers from the set $\{1,\dots, L\}$. They start in teams, and all agents in a team have the same starting node. The adversary decides the compositions of teams, and their starting nodes. Whenever an agent enters a node, it sees the entry port number and the states of all collocated agents; this information forms the input of the agent on the basis of which it transits to the next state and decides the current action. The aim is for all agents to gather at the same node and stop. Gathering is feasible, if this task can be accomplished for any decisions of the adversary, and its time is the worst-case number of rounds from the start till gathering.
 
 We consider the feasibility and time complexity of gathering teams of agents, and give a complete solution of this problem. It turns out that both feasibility and complexity of gathering depend on the sizes of teams. We first  concentrate on the case when all teams have the same size $x$. For the oriented line, gathering is impossible if $x=1$, and it can be accomplished in time $O(D)$, for $x>1$, where $D$ is the distance between the starting nodes of the most distant teams. This complexity is of course optimal. For the unoriented line, the situation is different. For $x=1$, gathering is also impossible, but for $x=2$, the optimal time of gathering is $\Theta(D\log L)$, and for $x\geq 3$, the optimal time of gathering is $\Theta(D)$.
 In the case when there are teams of different sizes, we show that gathering is always possible in time $O(D)$, even for the unoriented line. This complexity is of course optimal.

\vspace{2ex}

\end{abstract}

\vfill

\vfill

\thispagestyle{empty}
\setcounter{page}{0}
\pagebreak

\section{Introduction}\label{sec1}

Several mobile agents, navigating in a network, have to meet at the same node. This basic task, known as gathering, has been thoroughly studied in the literature. It has many applications, both in everyday life and in computer science. Mobile agents can be people who have to meet in a city whose streets form a network. In computer science applications, agents can represent either software agents in computer networks, or mobile robots in networks formed by, e.g.,  corridors of a contaminated mine, where the access is dangerous for humans.
The purpose of gathering may be to exchange data previously collected by software agents that accessed a distributed database, or to coordinate some future task of mobile robots, such as network decontamination, based on previously collected samples of the ground or of air. 
Since, in view of cost reasons, mobile agents are often simple devices with limited memory, sensor capacity and computation power, it is natural to model them as deterministic automata that can interact only when they meet. 

We consider the feasibility and time complexity of gathering mobile agents, modelled as automata, in one of the simplest networks that is the infinite line, i.e. the infinite graph all of whose nodes have degree 2. While the gathering problem in networks has a long research history, mobile agents in this context were usually modelled either as machines with unbounded memory or it was assumed that agents can ``see'' other agents at a distance. To the best of our knowledge, this classic problem has never been investigated before for agents modelled as deterministic automata with bounded memory, navigating in an infinite environment and interacting with other agents only at meetings.\footnote{It could seem that the recent work \cite{PP2} is a counterexample to this rule. However, the focus of  \cite{PP2} is on computing functions by teams of automata, and the environment in which agents navigate is the rooted oriented half-line. In this graph, gathering is trivial: all agents go toward the root and stop.}

\subsection{The model}
The environment in which agents navigate is modelled as an infinite line, i.e., the infinite graph all of whose nodes have degree 2. Nodes of the line do not have labels, and ports at each node are labeled $-1$ and $1$. We consider two variations of this port-labeled graph: in the {\em oriented line}, ports corresponding to any edge at both extremities of it are different; in the {\em unoriented line}, port labeling at any node is arbitrary.

There are several agents navigating in a line. They are
modelled as deterministic Mealy automata. Agents move in synchronous rounds. 
All agents start in the same round 0.
In each round, an agent can choose one of the following actions: take the port $-1$, take the port 1, or stay idle.
Agents have distinct labels which are integers from the set $\{1,\dots, L\}$. 
Agents start in $R$ teams, and the total number of agents is at most $L$. Agents in a team have the same starting node, called the {\em base} of the team. The sizes $x_{max}$ and $x_{min}$ of the largest and smallest teams,
the size $L$ of the space of labels, and the total number $R$ of teams are known in advance and can be used in the design of the agents. Throughout this article, we assume that $R>1$. If there is only one team, gathering is trivially achieved in the wake-up round.
The adversary knows all the agents and decides the composition of teams and their bases. In the case of the unoriented line, the adversary also decides the port labeling at each node.

When entering a node, an agent sees the entry port number, and 
sees the set of states of all currently collocated agents. If in the current round the agent stays idle, it also sees this set of states. (This is how agents interact). This information, together with the size $L$ of the label space and  the number $R$ of teams, forms the input of the agent in a given round. ($L$ and $R$ form the fixed part of the input that never changes).
Based on its state and this input, the agent transits to some state, and produces an output action belonging to the set $\{-1,1,0\}$, to be executed in the next round. Output action $-1$ means taking port $-1$, output action $1$ means taking port $1$, and output action $0$ means staying idle in the next round. 
The sequence of the above actions forms the {\em trajectory} of an agent. 
Each agent has a special state STOP, upon transition to which, it stays idle forever.

We now formalize the above intuitive description. All agents are represented by the same deterministic Mealy automaton
${\cA}=({\cal I},O,Q,\delta,\lambda$).
The agent with label $\ell\in \{1,\dots ,L\}$ has the set of states $Q^{\ell}$, where all sets $Q^{\ell}$ are pairwise disjoint. Thus the label of an agent can be recognized from each of its states.
$Q=Q^{1} \cup \cdots \cup Q^{L}$ denotes the union of sets of states of all the agents.
The size of $Q$, i.e., the number of states of the automaton, is denoted by $H$. It is a function of the size $L$ of the label space.
Let $\cal Q$ be the set of all subsets of $Q$. 
The common input alphabet for all agents is the set ${\cal I}=\{(L,R)\}\times\{-1,1,0\} \times {\cal Q} $. This corresponds to the intuition that an input $I\in {\cal I}$ is a triple whose first term is the {\em fixed part of the input} consisting of the pair of integers $(L,R)$,  whose second term is the {\em move part of the input} which is either the entry port number in the current round or 0 if the agent is idle in the current round (it is also 0 in round 0), and whose third term is the {\em states part of the input} which is the set of states of other currently collocated agents. Note that this set may be empty, if the agent is currently alone. The common output alphabet $O=\{-1,1,0\}$ corresponds to the output actions intuitively described above.

It remains to define the transition function $\delta$ and the output function $\lambda$.
The transition function $\delta:Q \times {\cal I} \to Q$ takes the current state of an agent and its current input, and produces the state to which this agent transits in the next round. The fact that $Q^{\ell}$ is the set of states of agent $\ell$ is reflected by the following restriction: 
if $q\in Q^{\ell}$ then $\delta(q, I) \in Q^{\ell}$, for any input $I$, i.e., under any input, the transition function maps a state of a given agent to a state of the same agent.
The output function $\lambda:Q \times {\cal I} \to O$ takes the current state of an agent and its current input, and produces the output action to be executed in the next round.

According to the established custom in the literature on automata navigating in graphs, we present the behavior of our automata by designing procedures that need only remember a bounded number of bits, and thus can be executed by deterministic automata, rather than formally describing the construction of a Mealy automaton by defining its output and state transition functions.

The aim is for all agents to gather at the same node and transit to state STOP. The first round in which all the agents are at the same node in state STOP is called the {\em gathering round}.  Gathering is feasible, if this task can be accomplished for any decisions of the adversary, and its time is the worst-case (largest) gathering round, over all decisions of the adversary.


\subsection{Our results}

We give a complete solution of the problem of feasibility and time complexity of gathering teams of agents on the infinite line. 
It turns out that both feasibility and complexity of gathering depend on the sizes of teams. We first  concentrate on the case $x_{max}=x_{min}=x$, i.e.,  when all teams have the same size $x$. For the oriented line, gathering is impossible if $x=1$, and it can be accomplished in time $O(D)$, for any $x>1$, where $D$ is the distance between the bases of the most distant teams. The first fact means that, for $x=1$ and for arbitrary agents, the adversary can place them in the oriented line (at distinct nodes) in such a way that they will never gather. The second fact means that, for any $x>1$ and any $R>1$, there exist $xR$ agents, each assigned a distinct label from
$\{1,\dots,L\}$, where $L\geq xR$, such that if the adversary composes them in arbitrary $R$ teams of size $x$ and selects the $R$ bases of the teams as arbitrary nodes of the oriented line with the most distant nodes at distance $D$, then the agents will gather in time $O(D)$.
This complexity is of course optimal. 

For the unoriented line, the situation is different. For $x=1$, gathering is also impossible which means that, for $x=1$ and for arbitrary agents,  the adversary can choose a port labeling of the line, and can place the agents at distinct nodes in such a way that they will never gather. This directly follows from the above impossibility result on the oriented line, as the adversary can choose the port labeling as in the oriented line. 

However, 
for $x=2$, the optimal time of gathering in the unoriented line turns out to be $\Theta(D\log L)$. To show this we prove two facts. First, we show that 
there exist $2R$ agents, each assigned a distinct label from
$\{1,\dots,L\}$, where $L\geq 2R$,
with the property that if the adversary chooses an arbitrary port labeling of the line, composes the agents in arbitrary $R$ teams of size $2$ and selects the $R$ bases of the teams as arbitrary nodes of the line with the most distant nodes at distance $D$, then the agents will gather in time $O(D\log L)$. Second, we prove that this complexity is optimal, even for two teams of agents, each of size 2.
In fact, we show that the ``difficult'' port labeling of the line can be chosen the same for any agents: it is the {\em homogeneous} port labeling in which, for every edge $e$, port numbers at both extremities of $e$ are equal.  
More precisely, we show that for any agents, the adversary can select two teams of agents of size 2 and choose bases of these teams as nodes at an arbitrarily large distance $D$ on the line with homogeneous port labeling, so that the gathering time will be at least   $cD\log L$, for some constant $c$. This shows that the complexity $O(D\log L)$ is tight.

We also prove that, for any $x>2$, the optimal time of gathering in the unoriented line is $\Theta(D)$.
This means that, for any $x>2$ and any $R>1$, there exist $xR$ agents with the property that if the adversary arbitrarily chooses the port labeling of the line,  composes the agents in arbitrary $R$ teams of size $x$ and selects the $R$ bases of the teams as arbitrary nodes with the most distant nodes at distance $D$, then the agents will gather in time $O(D)$. This complexity is of course optimal. 

The above described results show the solutions in the case when all teams of agents have the same size $x$. To complete the picture, we show that,
in the case when $x_{max}>x_{min}$, i.e., when there are teams of different sizes, gathering is always possible in time $O(D)$, even for the unoriented line. This complexity is of course optimal. 

Solving the gathering problem for agents that are deterministic automata navigating in an infinite environment requires new methodological tools. Traditional gathering techniques in graphs are {\em count driven}: agents make decisions based on counting steps. Since distances between agents may be unbounded, agents have to count unbounded numbers of steps.
When agents are modelled by automata, counting unbounded numbers of steps is impossible\footnote{This is the reason for our impossibility result in case of teams of size 1.}, hence we must use different methods. In all our gathering algorithms, changes of the agents' behavior are triggered not by counting steps but by events which are meetings between agents during which they interact. Hence our new technique is {\em event driven}. Designing the behavior of the agents based on meeting events, so as to guarantee gathering regardless of the adversary's decisions is our main methodological contribution.


\subsection{Related work}

Gathering mobile agents in graphs, also called {\em rendezvous} if there are only two agents, is a well studied topic in the distributed computing literature. Many scenarios have been adopted, concerning the capabilities and the behavior of the agents.

In the majority of the works on gathering, it is assumed that nodes do not have distinct identities, and agents cannot mark nodes: we follow these assumptions in the present article. However, departures from this model exist:
rendezvous was also considered in graphs whose nodes are labeled \cite{CCGKM,MP}, or when marking nodes by agents using tokens is allowed \cite{KKSS}.
For a survey of  randomized rendezvous we refer to the book
\cite{alpern02b}. 
Deterministic rendezvous in graphs was surveyed in \cite{Pe2}.

Most of the literature on rendezvous considered finite graphs and assumed the synchronous scenario, where
agents move in rounds.
In \cite{DFKP}, it was assumed that agents have distinct identities, and the  authors studied rendezvous time in trees, rings and arbitrary graphs.
In particular, they showed a rendezvous algorithm working in rings in time $O(D\log L)$, where $D$ is the initial distance between the agents and $L$ is the size of the label space. The algorithm from \cite{DFKP}, based on agents' labels transformation, was an inspiration of our procedure {\tt Dance}.

In \cite{TSZ07}, the authors presented rendezvous algorithms with time polynomial in the size of the graph and the length of agents' labels.
Gathering many agents in the presence of Byzantine agents that can behave arbitrarily was discussed in \cite{BDD,BDL}.
In this case only gathering among fault-free agents is required but Byzantine agents make it more difficult.

In the above cited works, agents are modeled as Turing machines and their memory is unbounded. Other studies concern 
the minimum amount of memory that agents must have in order to accomplish rendezvous \cite{CKP,FP}. In this case, agents are modeled as state machines and the number of states is a function of the size of graphs in which they operate.


%

Several authors studied synchronous rendezvous in infinite graphs. In all cases, agents had distinct identities and were modeled as Turing machines.
In \cite{CCGKM}, the authors considered rendezvous in infinite trees and grids,  under the assumption that the agents know their location in the graph (then the initial location can serve as a label). 
Rendezvous in the oriented grid was investigated heavily under various scenarios \cite{BCGIL,BBDDP,CCGKM}. 

In several works, asynchronous gathering was studied in the plane \cite{CFPS,fpsw} and 
in graphs
\cite{BBDDP,BCGIL,DPV,Stach}. In the plane, agents are modeled as moving points, and it is usually assumed that they can see the positions of other agents.
In graphs, an agent chooses the edge to traverse, but the adversary controls the speed of the agent. Then rendezvous
at a node cannot be guaranteed even in the two-node graph,  hence the agents are permitted to meet inside an edge.
For asynchronous rendezvous, the optimization criterion is the cost, i.e.,  the total number of edge traversals.
In \cite{BCGIL}, the authors designed almost optimal algorithms for asynchronous rendezvous in infinite multidimensional grids, assuming that an agent knows its position in the grid. In \cite{BBDDP,DPV, Stach} this assumption was replaced by a weaker assumption that agents have distinct identities. In \cite{BBDDP}, a polynomial-cost algorithm was designed for the infinite oriented two-dimensional grid. \cite{Stach} considered asynchronous rendezvous in the infinite line and  \cite{DPV} -- in arbitrary finite graphs. Tradeoffs between cost of  asynchronous gathering and memory size of the agents operating in arbitrary trees were investigated in \cite{BIOKM}.

\section{Preliminaries}\label{sec:prelim}

For any port labeling of a line, we will use two notions describing the two directions of the line.  For any node, we define the {\em plus direction} and the {\em minus direction} as the directions corresponding to port $1$ and port $-1$, respectively.
For the oriented line, these directions do not depend on the node, and for an arbitrary port labeling, they do. Moreover, we define the {\em left} and {\em right} direction of any line. For the oriented line, the left (resp. right) direction is the minus (resp. plus) direction. For any other port labeling, these directions are chosen arbitrarily. Hence, in the oriented line,
agents can identify directions left and right. For an arbitrary port labeling,  agents cannot identify them, so we will use these expressions only in comments and in the analysis. For any two nodes $v$ and $v'$ in a line, we define their distance $dist(v,v')$ as the number of edges between them.

Agents are identified with their labels. For any port labeling of a line, and for any base of a single agent navigating in the line, we define the {\em trajectory} of this agent as the infinite sequence of terms $-1,0,1$ with the following meaning. The $i$th term of the trajectory is 0 if the agent stays put in the $i$th  round, it is $-1$ if the agent takes port $-1$ in the $i$th round,  and  it is 1 if the agent takes port 1 in the $i$th round.

We say that a trajectory of an agent has a {\em period} of length $\tau$ if there exists an integer $t \geq 1$, such that, for any fixed $0\leq i <\tau$, the $(t+j\tau+i)$th term of the trajectory is the same for all $j\geq 0$.
The sequence of terms $-1,0,1$ corresponding to indices
$t+j\tau,t+j\tau+1,\dots,t+j\tau +\tau-1$ is called a period of this trajectory, and the sequence of terms $-1,0,1$ corresponding to indices $1,\dots ,t-1$ is called a {\em prefix} corresponding to this period.
A trajectory that has a period is called {\em periodic}.

Apart from the port labeling that yields the oriented line (port numbers at both extremities of each edge are different), we consider 
another important port labeling, called 
{\em homogeneous}, in which, for every edge $e$, port numbers at both extremities of $e$ are equal.  A line with this port labeling will be called homogeneous.

\begin{proposition}\label{periodic}
	The trajectory of any agent navigating either in the oriented or in the homogeneous line and starting at any node of it is periodic.
\end{proposition} 

\begin{proof}
	The proof is the same for the oriented and for the homogeneous line.
	Consider the trajectory of any agent starting at any node of one of these lines. 
	The number of states of the agent is bounded and, for a single agent,  there are three possible inputs in each round, 
	corresponding to three possible move parts of the input,
	$-1$, 1 or 0 (the other parts are fixed for a single agent). 
	Hence, there must exist two rounds $t_2>t_1$ such that in both of them the agent is in the same state and has the same input. Both in the oriented and in the homogeneous line, the port number $q$ that an agent takes in some round determines the entry port number $q'$ at the adjacent node ($q'=-q$ in the case of the oriented line, and $q'=q$ in the case of the homogeneous line).
	Let $\tau=t_2-t_1$. By induction on $i<\tau$ we conclude that,  
	for any $0\leq i <\tau$, the $t_1+(j\tau+i)$th term of the trajectory is the same for any $j\geq 0$. Hence the trajectory has a period of length $\tau$.
\end{proof}

The notion of a period permits us to define three important notions concerning the trajectory of an agent navigating in the oriented or in the homogeneous line: boundedness, the progress direction and the speed. We first define these notions for an agent in the oriented line. Consider any starting node (base) of the agent, and consider a period of length $\tau$ of its trajectory $\alpha$. Let $v$ and $v'$ be nodes where the agent is situated at the beginning of two consecutive periods. There are three cases: $v'$ is left of $v$, $v'=v$, and $v'$ is right of $v$. For the oriented line, this is independent of the base of the agent. It is easy to see that:
\begin{enumerate}[label=(\roman*)]
	\item 
	if $v'$ is left (resp. right) of $v$ then the agent visits all nodes of the line left (resp. right) of the base and only a bounded number $r(\alpha)$ (resp. $l(\alpha)$) of nodes right (resp. left) of the base
	($r(\alpha)$ and $l(\alpha)$ are independent of the base);
	in this case we say that the trajectory of the agent is {\em left-progressing} (resp. {\em right-progressing});
	\item
	if $v'=v$ then the agent visits only a bounded number of nodes: at most $b(\alpha)$ nodes on each side of the base ($b(\alpha)$ is independent of the base); in this case we say that the trajectory of the agent is {\em bounded};
\end{enumerate}
If the trajectory of the agent is left- or right-progressing, its {\em speed} is defined as $dist(v,v')/\tau$.
Notice that the definition of the speed does not depend on the choice of the period of the trajectory.

In the case of the homogeneous line, the situation is slightly different for two reasons. First, the progress direction
of a trajectory depends on the base.  Hence we will use notions {\em minus progressing} and {\em plus progressing}, where the directions are defined with respect to the base.
Second (unlike in the oriented line), it is possible that the node $v'$ at the beginning of a period is different from the node $v$ at the beginning of the preceding period but after two consecutive periods these nodes are the same. This happens, e.g., when the period is
$(1,1,-1)$. Hence, in order to define  progress direction and boundedness properly, we
consider {\em double periods}: a period $Q$ is double, if it is the concatenation $P*P$, for some period $P$.
For double periods, the above issue disappears. 

Consider any starting node $u$ (base) of the agent, and consider a double period of length $\tau$ of its trajectory $\alpha$. Let $v$ and $v'$ be nodes where the agent is situated at the beginning of two consecutive periods. We use the plus and minus direction with respect to $u$.
There are three cases: $v'$ is in minus direction from $v$, $v'=v$, and $v'$ is  in plus direction from $v$. Now we have:

\begin{enumerate}[label=(\roman*)]
	\item 
	if $v'$ is in minus (resp. plus) direction from $v$ then the agent visits all nodes of the line in minus (resp. plus) direction from $u$ and only a bounded number $r(\alpha)$ (resp. $l(\alpha)$) of nodes in plus (resp. minus) direction from $u$
	($r(\alpha)$ and $l(\alpha)$ are independent of the base);
	in this case we say that the trajectory of the agent is {\em minus-progressing} (resp. {\em plus-progressing});
	\item
	if $v'=v$ then the agent visits only a bounded number of nodes: at most $b(\alpha)$ nodes on each side of the base ($b(\alpha)$ is independent of the base); in this case we say that the trajectory of the agent is {\em bounded};
\end{enumerate}
If the trajectory of the agent is minus- or plus-progressing, its {\em speed} is defined as $dist(v,v')/\tau$, for any double period.
Similarly as before, the definition of the speed does not depend on the choice of the double period of the trajectory.
Notice that in the case of the homogeneous line, for a fixed base of an agent, we can still use the expression ``left-'' or ``right-progressing'' with respect to its trajectory, although, unlike for the oriented line, these notions also depend on the base and not only on the trajectory.

The following proposition  bounds the length of  a prefix and of a period of trajectories. 

\begin{proposition}\label{pi plus tau}
	For any periodic trajectory $\alpha$, denote by $\sigma$ the length of the shortest period of $\alpha$ and 
	by $\pi$ the length of the prefix preceding the first period of length $\sigma$. Then $\pi+\sigma \leq 3H$.
\end{proposition}

\begin{proof}
	Observe that there are three possible inputs (corresponding to move parts $-1$, 1, 0 of the input) and the agent can be in less than $H$ possible states other than STOP. Hence, there are two rounds $t_1<t_2 \leq 3H$, where the agent is in the same state and has the same input. Since $t_2$ is after the end of some period, round $t_2$ must be at least $\pi+\sigma$ which concludes the proof. 
\end{proof}

We will use the following fact that holds both for the oriented and for the homogeneous line.
Consider two agents $a$ and $b$ with bases $u$ and $v$, respectively.
We say that agent $b$ {\em follows} agent $a$, if either the trajectories of both of them are left-progressing and $u$ is left of $v$, or the trajectories of both of them are right-progressing and $u$ is right of $v$. The following proposition says intuitively that, if the initial distance  between the agents is sufficiently large, then:\\
(i) if the follower is not faster than the followed agent then agents can never meet, and\\
(ii) if the difference between the speeds of the agents is small then they cannot meet soon.

\begin{proposition}\label{follow}
	Consider two agents, $a$ and $b$, with bases $u$ and $v$, respectively, either in the oriented line or in the homogeneous line, such that $D=dist(u, v)>72H^2+6H$ and agent $b$ follows agent $a$.
	Let $z=V(b)-V(a)$, where $V(a)$ and $V(b)$ denote the speeds of the trajectories of $a$ and $b$, respectively.\\
	i) If $z\le 0$, then agents $a$ and $b$ can never meet;\\ 
	ii) If $z>0$, then agents $a$ and $b$ cannot meet before round $3H+36H^2+\lceil (D-6H-72H^2)/z\rceil$.
\end{proposition}

\begin{proof}
	Without loss of generality, assume that both agents $a$ and $b$ have left-progressing trajectories $\alpha(a)$ and $\alpha(b)$.
	Node $u$ is left of $v$, in view of the assumption.

	Let $\tau(a)$ be the length of the shortest double period of the trajectory $\alpha(a)$ and let
	$\tau(b)$ be the length of the shortest double period of the trajectory $\alpha(b)$. In view of Proposition \ref{pi plus tau}, we have
	$\tau(a)\leq 6H$  and $\tau(b)\leq 6H$.
	Each of the  trajectories $\alpha(a)$ and $\alpha(b)$ has a period of length
	$\tau=\tau(a) \cdot \tau (b)$. Hence $\tau \leq 36H^2$. Let $\pi(a)$ be the length of the prefix preceding the first period of length $\tau$ in the trajectory $\alpha(a)$, and let $\pi(b)$ be the length of the prefix preceding the first period of length $\tau$ in the trajectory $\alpha(b)$. Without loss of generality, suppose that $\pi(b) \geq \pi(a)$. The segment of the trajectory $\alpha(a)$ corresponding to rounds 
	$\pi(b)+1,\dots, \pi(b) +\tau$ is also a period of $\alpha(a)$. Hence both trajectories $\alpha(a)$ and $\alpha(b)$ can be viewed as having a prefix of length $\pi(b)$ and a double period of length $\tau$. In view of Proposition \ref{pi plus tau}, we have 
	$\pi(b)\leq 3H$. In any round $t\leq \pi(b)$, the agents can decrease their initial distance by at most 2, hence they cannot meet by round $\pi(b)$.
	
	Let $\delta_i$, for $i \geq 1$, denote the distance between the agents at the beginning of the $i$-th double period of length $\tau$ of their trajectories. Since, $\pi(b)\leq 3H$, the agents can decrease their initial distance until round  $\pi(b)$ by at most $2 \pi(b) \leq 6H$, and hence $\delta_1 \ge D-2\pi(b)>72H^2+6H-6H=72H^2$. 
	
	Case 1. $z\le 0$.
	In this case, $V(a)\geq V(b)$. This implies that $\delta_i$ cannot increase as $i$ grows, and hence $\delta_i >72H^2$, for any $i\geq 1$. During any double period of length $\tau$, the agents can decrease the distance separating them at the beginning of this period by at most $2\tau \leq 72H^2$. Hence, in any round $t>\pi(b)$, their distance is strictly positive which implies that they can never meet.
	
	Case 2. $z>0$. Observe that $\delta_i-\delta_{i+1}=z\cdot\tau$ for $i\ge 1$.
	As $\delta_1\ge D-2\pi(b)\ge D-6H$, we have $\delta_i\ge D-6H-z\cdot\tau\cdot (i-1)>72H^2$, for $i<1+(D-6H-72H^2)/(z\cdot \tau)$.
	Notice that if $i<1+(D-6H-72H^2)/(z\cdot \tau)$, then the $i$-th double period of length $\tau$ of their trajectories ends in round 
	\begin{align*}
		\pi(b)+\tau\cdot i&<\pi(b)+\tau(1+(D-6H-72H^2)/(z\cdot \tau))\\
		&=\pi(b)+\tau+(D-6H-72H^2)/z \\
		&\le 3H+36H^2+\lceil (D-6H-72H^2)/z \rceil.
	\end{align*}
	
	During any double period of length $\tau$, the agents can decrease the distance separating them at the beginning of this period by at most $2\tau \leq 72H^2$.
	Hence, in any round $\pi(b)<t< 3H+36H^2+\lceil (D-6H-72H^2)/z \rceil$, the distance between the agents is strictly positive.
	Hence, they can not meet before round $3H+36H^2+\lceil (D-6H-72H^2)/z\rceil$.
\end{proof}

\section{The Impossibility Result}\label{sec:impos}

In this section, we show that if all teams are of size  $x=1$, then gathering is impossible for some port labeling of the line.
This port labeling is that of the oriented line.

\begin{theorem}\label{impos}
	Consider an arbitrary set of agents. Then the adversary can place these agents at distinct nodes of the oriented line in such a way that no pair of agents will ever meet.
\end{theorem}

\begin{proof}
	Partition the set $T$ of  trajectories of all considered agents into three sets: the set $A$ of left-progressing trajectories, the set $B$ of bounded trajectories, and the set $C$ of right-progressing trajectories. Define $M$ as the largest among the following integers: $36H^2+3H$, the maximum of $r(\alpha)$ over all $\alpha\in A$, the maximum of $b(\alpha)$ over all $\alpha\in B$, and
	the maximum of $l(\alpha)$ over all $\alpha\in C$.
	Recall that $r(\alpha)$, $b(\alpha)$ and $l(\alpha)$ were defined in Section \ref{sec:prelim}.
	
	Order all trajectories in $A$ in non-increasing order $(\alpha_1,\dots, \alpha_k)$ of their speeds.
	Order all trajectories $(\beta_1,\dots, \beta_m)$ in $B$ arbitrarily.
	Order all trajectories in $C$ in non-decreasing order $(\gamma_1,\dots, \gamma_n)$ of their speeds.
	Choose a sequence of nodes $(u_1,\dots,u_k, v_1,\dots,v_m,w_1,\dots,w_n)$ from left to right, in such a way that consecutive nodes are at distance $2M+1$. Place the base of the agent with trajectory $\alpha_i$ at node $u_i$, for $1\leq i \leq k$,
	place the base of the agent with trajectory $\beta_i$ at node $v_i$, for $1\leq i \leq m$,
	and place the base of the agent with trajectory $\gamma_i$ at node $w_i$, for $1\leq i \leq n$.
	
	We claim that, for this choice of bases, no meeting can occur between any agents.
	Consider any pair of agents placed as above and consider their trajectories $\xi_1, \xi_2$.
	First observe that if at least one of these trajectories is from the set $B$ then the sets of nodes visited by both agents are disjoint, in view of the definition of $M$, and hence the agents cannot meet. The same argument holds if one of the trajectories is from the set $A$ and the other is from the set $C$.
	Hence, we may assume that either both trajectories are from $A$ or both trajectories are from $C$.
	We give the argument in the first case; the other is similar. Suppose, without loss of generality, that $\xi_1=\alpha_i$ and $\xi_2=\alpha_j$, for $i<j$. Then the agent $a_j$ with trajectory $\alpha_j$ follows the agent $a_i$ with trajectory $\alpha_i$. The speed of $\alpha_j$ is not larger than the speed of $\alpha_i$. In view of the definition of $M$ and of Proposition \ref{follow} (i), the agents never meet.
\end{proof}

\section{The Oriented Line}\label{sec:oriented}

In this section, we consider gathering of $R$ teams of agents of equal size $x$ in the oriented line.
In view of Theorem \ref{impos}, if the size $x$ of each team is one, then the adversary can prevent gathering.
So, we assume that $x>1$, throughout this section.
The algorithm is presented  in Section \ref{sect-oriented-alg}, while the proof of its correctness and complexity is deferred to Section \ref{sect-oriented-correctness}.

\subsection{The Algorithm}
\label{sect-oriented-alg}

We now describe Algorithm {\tt Oriented} that accomplishes gathering of arbitrary $R$ teams of $x>1$ agents in the oriented line, in optimal time $O(D)$, where $D$ is the distance between the bases of the most distant teams.

\noindent
{\bf Modes.} In each round of the algorithm execution, each agent is in some mode, encoded in the state of the agent. In each mode, an agent performs some action. 
All modes and their corresponding actions, used in Algorithm {\tt Oriented}, are listed in Table \ref{tab-modes-oriented-simple}.

\begin{table*}[t]
	\centering
	\caption{\label{tab-modes-oriented-simple}The four modes used in Algorithm {\tt Oriented} and their corresponding actions. }
	\begin{tabular}{ |c| c| } 
		\hline
		Mode & Action		\\
		\hline
		\hline
		$\rightSlow$ & Move one step right, stay put for one round, and repeat\\
		\hline
		$\rightFast$ & Move one step right in each round\\
		\hline
		$\leftFast$ & Move one step left in each round\\
		\hline
		$\senti$ & 	Stay put \\
		\hline
	\end{tabular}
\end{table*}

\noindent
\textbf{The high-level idea of the algorithm.}
After waking up, the agent with the smallest label in each team assigns itself mode $\rightSlow$ which means that it moves right with speed one-half, and the other $x-1$ agents in the team assign themselves mode $\senti$ and stay put at their base.
An agent $a$ in mode $\rightSlow$ (except the agent from the rightmost team) will meet agents in mode $\senti$.
Each time such a meeting happens, agent $a$ counts the total number of agents in mode $\senti$ it has seen so far.
Notice that among all the agents in mode $\rightSlow$, only the agent $b$ from the leftmost team will meet $(x-1)(R-1)$ agents in mode $\senti$ (except $\senti$ agents from its own team).
After meeting $(x-1)(R-1)$ agents in mode $\senti$, agent $b$ switches to mode $\rightFast$ which means that it moves right with speed 1.
It is the only agent that ever switches to mode $\rightFast$.
Speed 1 allows it to meet agents in mode $\rightSlow$ that it follows.
Each time agent $b$ meets an agent in mode $\rightSlow$, agent $b$ counts the total number of agents in mode $\rightSlow$ it has seen so far, and the agent in mode $\rightSlow$ switches to mode $\senti$.
After meeting $R-1$ agents in mode $\rightSlow$, agent  $b$ switches to mode $\leftFast$. 
Let $c$ be the $(R-1)th$ agent in mode $\rightSlow$ that agent $b$ met. At this meeting, agent $c$ also switches to mode $\leftFast$.

So, both agents $b$ and $c$ move together left  with speed 1.
At this time, all the other $(x\cdot R-2)$ agents are in mode $\senti$, left of agents $b$ and $c$.
Then, each time a meeting between agents in mode $\leftFast$ and agents in mode $\senti$ happens, all the agents in mode $\senti$ switch to mode $\leftFast$ and move together with the other agents in mode $\leftFast$.
In the end, all $x\cdot R$ agents gather at the base $u$ of the leftmost team. 
They know that this happened, by counting the number of agents at $u$,
and transit to state {\tt STOP}.

\noindent
\textbf{Detailed description of the algorithm.}
As explained above, in each round of the
algorithm, each agent is in some mode, and modes may change at meetings.
We regard meetings as events and there are four types of events in the algorithm.
In particular,
{\tt Event A} is the wake-up event; {\tt Event B} is a meeting between agents in modes $\rightSlow$ and $\senti$; {\tt Event C} is a meeting between agents in modes $\rightFast$ and $\rightSlow$; and {\tt Event D} is a meeting between agents in modes $\leftFast$ and $\senti$.
At each event, agents may change their modes which implies taking the action specified in Table \ref{tab-modes-oriented-simple}. 
Obviously, all four events are pairwise exclusive, hence the choice of the action is unambiguous. 
Algorithm {\tt Oriented} is described as a series of actions corresponding to these four events.
For an agent $a$ in mode $\rightSlow$, we define the variable $a.countSenti$ that stores the total number of agents in mode $\senti$ that it has seen.
For an agent $a$ in mode $\rightFast$, we define the variable $a.countSlow$ that stores the total number of agents in mode $\rightSlow$ that it has seen.
The details of Algorithm {\tt Oriented} are as follows:

\noindent
{\bf Event A:} All $R$ teams are woken up. \\

\noindent
$mov := $ the agent with smallest label in the team; \\
$mov$ assigns itself mode $\rightSlow$; \\
$mov.countSenti := 0$; \\
All agents in the team, other than $mov$, assign themselves mode $\senti$;\\

\noindent
{\bf Event B:} An agent $e$ in mode $\rightSlow$ meets $x-1$ agents in mode $\senti$.\\

\noindent
$e.countSenti := e.countSenti+(x-1);$\\
{\bf if} $e.countSenti = (x-1)(R-1)$ {\bf then}\\ 
\hspace*{1cm} $e$ switches to mode $\rightFast$; \\
\hspace*{1cm} $e.countSlow := 0$; \\

\noindent
{\bf Event C:} An agent $e$ in mode $\rightFast$ meets an agent $f$ in mode $\rightSlow$.\\

\noindent
$e.countSlow := e.countSlow+1;$\\
{\bf if} $e.countSlow = R-1$ {\bf then}\\ 
\hspace*{1cm} $e$ switches to mode $\leftFast$; \\
\hspace*{1cm} $f$ switches to mode $\leftFast$; \\

\noindent
{\bf Event D:} A set $G$ of agents in mode $\leftFast$ meet a set $G'$ of agents in mode $\senti$.\\

\noindent
{\bf if} $|G|+|G'|=xR$ {\bf then}\\ 
\hspace*{1cm} each agent transits to state {\tt STOP}; \\
{\bf else} \\
\hspace*{1cm} each $f\in G'$ switches to mode $\leftFast$; \\

This completes the detailed description of Algorithm {\tt Oriented}. 
We will have to prove that {\tt Events B,C,D} are the only types of meetings that can occur during the execution of the algorithm.

\subsection{Correctness and complexity}
\label{sect-oriented-correctness}
In this section, we present the proof of correctness of Algorithm {\tt Oriented}, and
establish its complexity.

\begin{theorem}
	Algorithm {\tt Oriented} gathers $R$ teams of $x>1$ agents each, in the oriented line in $7D$ rounds, where $D$ denotes the distance between the bases of the most distant teams.
\end{theorem}

\begin{proof}
	First, we prove the correctness of Algorithm {\tt Oriented}.
	Consider all $R$ teams on the oriented line from left to right.
	Let $a_i$ denote the smallest-labeled agent in the $i$-th team, for $1\le i\le R$.
	Notice that in the wake-up round, there are $(x-1)(R-i)$ agents in mode $\senti$, right of agent $a_i$, for each $1\le i\le R$.
	After arriving at the base of each team right of its own base, agent $a_i$ meets $(x-1)$ new agents in mode $\senti$ and increases its variable $a_i.countSenti$ by $(x-1)$.
	
	Let $t_1$ denote the round when agent $a_1$ arrives at the base of the rightmost team.
	By round $t_1$, agent $a_1$ has met all the $(R-1)(x-1)$ agents in mode $\senti$, right of its base, and the variable $a_1.countSenti$ has been increased to $(R-1)(x-1)$.
	So, $a_1$ switches to mode $\rightFast$ in round $t_1$, in view of the algorithm.
	Note that only agent $a_1$ will meet $(x-1)(R-1)$ agents in mode $\senti$, which means that it is the only agent that ever switches to mode $\rightFast$.
	As any two agents $a_i$ and $a_j$ in mode $\rightSlow$ have the same trajectory, if $a_i$ follows $a_j$, then the distance between both agents is always the same, so $a_i$ cannot meet $a_j$.
	In round $t_1$, the only agent in mode $\rightFast$, i.e., $a_1$, appears.
	In this round, there are $(R-i)$ agents in mode $\rightSlow$, right of agent $a_i$, for each $i\ge 1$.
	After round $t_1$, agent $a_1$ follows all the agents $a_2, \dots, a_{R}$ (that are in mode $\rightSlow$), and has speed larger than theirs.
	Hence, $a_1$ meets each of these agents at some point.
	Each time $a_1$ meets an agent in mode $\rightSlow$, the variable $a_1.countSlow$ is increased by one.
	
	Let $t_2$ denote the round when $a_1$ meets $a_R$.
	By round $t_2$, agent $a_1$ has met $(R-1)$ agents in mode $\rightSlow$, and the variable $a_1.countSlow$ has been increased to $(R-1)$.
	So, in round $t_2$, both agents $a_1$ and $a_R$ switch to mode $\leftFast$, in view of Algorithm {\tt Oriented}.
	Between round $t_1+1$ and $t_2-1$, each time agent $a_1$ meets an agent $a_i$, for $1<i<R$,  in mode $\rightSlow$, $a_i$ switches to mode $\senti$.
	Agent $a_1$ maintains its mode $\rightFast$ and leaves $a_i$ behind.
	In round $t_2$, all $(x\cdot R-2)$ agents, other than $a_1$ and $a_R$, are in mode $\senti$, left of $a_1$ and $a_R$.
	After round $t_2$, agents $a_1$ and $a_R$ move left with speed one.
	Each time they meet a set of agents in mode $\senti$, these agents in mode $\senti$ switch to mode $\leftFast$; as a result, 
	after any meeting after round $t_2$ at a node $w$, all the agents at $w$ move  together left with speed one.
	
	Let $t_3$ denote the round when $a_1$ reaches its base $u$ again.
	In round $t_3$, all the $x\cdot R$ agents gather at $u$ and transit to state STOP.
	Therefore, gathering is achieved in round $t_3$.
	
	In order to complete the proof of correctness, we show that no meetings other than {\tt Event B}, {\tt Event C}, or {\tt Event D} can happen.
	Before round $t_1$, each agent is in one of the modes $\senti$ or $\rightSlow$, so meetings can happen only between an agent in mode $\rightSlow$ and agents in mode $\senti$ (Agents in mode $\rightSlow$ cannot meet because they follow each other and have the same trajectory).
	Hence, only {\tt Event B} can be triggered before round $t_1$.
	Between rounds $t_1$ and $t_2-1$, none of the agents in mode $\senti$ is right of agent $a_1$, and none of the agents in mode $\rightSlow$ is left of agent $a_1$.
	Thus, meetings can happen only between an agent in mode $\rightFast$ and an agent in mode $\rightSlow$.
	Hence, only {\tt Event C} can be triggered between rounds $t_1$ and $t_2-1$.
	From round $t_2$ on, each agent is in one of the modes $\senti$ or $\leftFast$, and none of the agents in mode $\senti$ is right of agents in mode $\leftFast$.
	Thus, meetings can happen only between agents in mode $\leftFast$ and agents in mode $\senti$.
	Hence, only {\tt Event D} can be triggered after round $t_2$.
	Therefore, during the execution of Algorithm {\tt Oriented}, no meetings other than {\tt Event B}, {\tt Event C}, and {\tt Event D} can happen. 
	This completes the proof of the correctness of Algorithm {\tt Oriented}.
	
	It remains to analyze the complexity of Algorithm {\tt Oriented}.
	The distance between the base $u$ of the leftmost team and the base $v$ of the rightmost team is $D$.
	Agent $a_1$ moves with speed one half from $u$ to $v$.
	So, agent $a_1$ reaches $v$ in $2D$ rounds, hence $t_1=2D$.
	Recall that both agents $a_1$ and $a_R$ have the same trajectory in the first $t_1$ rounds.
	Initially, their distance is $D$; hence, their distance is also $D$ in round $t_1$.
	From round $t_1$ on, agent $a_1$ follows agent $a_R$ and moves twice as fast as agent $a_R$.
	More precisely, agent $a_1$ goes one step right every round, and agent $a_R$ goes one step right every second round.
	So, agent $a_1$ meets $a_R$ in additional $2D$ rounds, hence $t_2=4D$.
	In round $t_2$, the distance between agent $a_1$ and its base $u$ is $3D$.
	After round $t_2$, agent $a_1$ moves towards $u$ with speed one, so it reaches $u$ again in additional $3D$ rounds; hence, $t_3=7D$.
	Therefore, gathering happens in round $7D$.
	This concludes the proof of the theorem.
\end{proof}

\section{The Unoriented Line}\label{sec:unoriented}

In this section, we consider gathering teams of the same size $x$ in the unoriented line, in which the port labeling at each node is arbitrary. We will use the expressions ``left'' and ``right'', to denote the two directions of the line but we need to keep in mind that agents cannot identify these directions, so we will use these expressions only in comments and in the analysis. 

It follows from Section \ref{sec:impos} that, if the team size is $x=1$, then gathering is impossible for some port labeling, namely that of the oriented line. This section is devoted to showing that, in the unoriented line, the optimal gathering time is  $\Theta(D\log L)$, for $x=2$, and it is $\Theta(D)$, for $x>2$, where $D$ is the distance between the bases of the most distant teams.  We first consider the case $x=2$.

\subsection{Teams of size 2}

Results of this section are by far the most technically difficult. The section is organized as follows. First we prove that, regardless of the deterministic automaton used, and even for two teams of size 2, the adversary can choose a particular port labeling of the line, namely the homogeneous port labeling, and it can choose the bases at an arbitrarily large distance $D$ and a composition of teams in such a way that the time of gathering is at least $cD\log L$, for some positive constant $c$.  Then we show, for any number $R>1$ of teams of size 2,  an algorithm always guaranteeing gathering  in time $O(D\log L)$. In view of the above lower bound, this complexity is optimal.

\subsubsection{The lower bound}

Consider the homogeneous line. Let $C=\lfloor L/2\rfloor$.
We will consider the following $C$ teams of size 2: $\{1,2\}$, $\{3,4\}$, ... , $\{2C-1,2C\}$. For each of those teams $\{a,a+1\}$ there are fixed trajectories $\alpha(a)$ and $\alpha(a+1)$ corresponding to agents $a$ and $a+1$, respectively, yielded by the automaton used. Each of the above teams is called {\em one-way} if either both corresponding trajectories are plus-progressing, or both are minus-progressing, or at least one of them is bounded.

\begin{lemma}\label{excl1}
	If there exist two one-way teams $\{a,a+1\}$ and $\{b,b+1\}$, for odd $a,b \leq 2C-1$, then there exist arbitrarily large  integers $D$ such that the adversary can place these teams at two bases at distance $D$, so that gathering will never happen.
\end{lemma}

\begin{proof}
	Consider any two one-way teams $\{a,a+1\}$ and $\{b,b+1\}$, for odd $a,b \leq 2C-1$.
	Let $T$ be the set of the four trajectories of agents in these teams.
	Let $A$ be an upper bound on the number of nodes visited by any of the minus-progressing trajectories from $T$ in the plus-direction from its base.
	Let $B$ be an upper bound on the number of nodes visited by any of the plus-progressing trajectories from $T$ in the minus-direction from its base.
	Let $C$ be an upper bound on the number of nodes visited in either direction from its base by any of the bounded trajectories from $T$.
	Let $M=\max \{2A,2B,2C\}$.
	
	The adversary chooses as the base of the team $\{a,a+1\}$ a node $u$ such that all unbounded trajectories of agents in this team (there can be 0, 1, or 2 of them) are left-progressing. Let $v$ be the node at distance $M+1$ right of $u$ and let $v'$ be the node at distance $M+2$ right of $u$.
	Then the adversary chooses as the base of the team $\{b,b+1\}$ one of the nodes $v$ or $v'$ satisfying the condition that  
	all unbounded trajectories of agents in this team (there can be 0, 1, or 2 of them) are right-progressing. In view of the homogeneous port numbering of the line, one of the nodes $v$ or $v'$ must satisfy this condition. Since $dist(u,v)>M$ and $dist(u,v')>M$, the sets of nodes of all agents from different teams are pairwise disjoint, in view of the definition of $M$. Hence the agents of different teams can never meet.
\end{proof}

In view of Lemma \ref{excl1}, we may assume that there is at most one  one-way team. Hence there exist $C-1$ teams $\{a,a+1\}$, for odd $a \leq 2C-1$  that are not one-way. Call these teams {\em canonical}. By definition, for every canonical team, one of the corresponding trajectories is plus-progressing and the other is minus-progressing. We call the plus-progressing (resp. minus-progressing) trajectory
the {\em plus-trajectory} (resp. the {\em minus-trajectory}) of the team. The agents corresponding to these trajectories are called the {\em plus-agent} (resp. {\em minus-agent}) of the team.

Consider two agents $a$ and $b$ with bases $u$ and $v$, respectively, such that $u$ is left of $v$. We say that the agents are {\em diverging}, if the trajectory of $a$ is left-progressing and the trajectory of $b$ is right-progressing. 

\begin{lemma}\label{lem-diverging}
	If the distance between the bases $u$ and $v$ of diverging agents is larger than $6H$ then these agents can never meet.
\end{lemma}

\begin{proof}
	Without loss of generality, assume $u$ is left of $v$.
	Let $u'$ denote the node at distance $3H$ of $u$, right of $u$.
	Let $v'$ denote the node at distance $3H$ of $v$, left of $v$.
	As $dist(u, v)>6 H$, it follows that $u'$ is left of $v'$.
	As the agent based at $u$ has a left-progressing trajectory, it cannot reach any nodes that are right of $u'$, in view of Proposition \ref{pi plus tau}.
	As the agent based at $v$ has a right-progressing trajectory, it cannot reach any nodes that are left of $v'$, in view of Proposition \ref{pi plus tau}.
	Therefore, no meeting can occur between the agents.
\end{proof}

We will use the following lemma that is a direct consequence of Theorem 3.1 from \cite{DFKP}.\footnote{Theorem 3.1 from \cite{DFKP} holds even in a ring and even if agents are Turing machines.}

\begin{lemma}\label{dfkp}
	Consider a set of $g$ agents, each of which is assigned a trajectory. There exist two agents in this set, such that if they start at two nodes of a line at distance $D$ and follow their trajectories then the first meeting between them occurs after at least $\frac{1}{6} D \log g$ rounds.
	
\end{lemma}

The main result of this subsection is the following theorem.

\begin{theorem}
	For any  automaton formalizing the agents, there exists a positive constant $c$ such that, for arbitrarily large $D$, the adversary can choose two canonical teams and their bases at distance $D$ on the homogeneous line, so that gathering takes time at least
	$cD\log L$. 
\end{theorem}

\begin{proof}
	Let $p=\lfloor L^{1/3}\rfloor$.
	We consider the partition of the interval $(0,1]$ into subintervals $I_1,\dots,I_p$, where $I_i=(\frac{i-1}{p}, \frac{i}{p}]$, for $i=1, \dots ,p$.
	For any $i, j \in \{1,\dots, p\}$, denote $\Sigma_{i,j}= I_i \times I_j$. Hence, $\Sigma_{i,j}$ form a partition of the square $(0,1]\times (0,1]$ into $p^2$ squares. We assign each canonical team to one of these squares as follows: a canonical team is assigned to square $\Sigma_{i,j}$, if the speed of the minus-trajectory of the team is in $I_i$ and the speed of the plus-trajectory of the team is in $I_j$. Since there are $p^2$ squares and $C-1=\lfloor L/2\rfloor-1$ canonical teams, there is at least one square $\Sigma_{i,j}$ to which at least $g=p/2$ canonical teams are assigned. Let $\Sigma=\Sigma_{i,j}$ denote any such square (there can be many of them).
	
	We now describe the decisions of the adversary. Let $D$ be any odd integer larger than $2\cdot(72H^2+6H)$. There are two cases.
	
	Case 1. $i \geq j$.
	
	Consider the plus-agents of the teams assigned to $\Sigma$. There are at least $g$ of them.
	By Lemma \ref{dfkp}, there exist two of them, $p_1$ and $p_2$,  such that  if they are placed at any two nodes at distance $D$ and follow their trajectories, then they cannot meet before $\frac{1}{6}D\log g$ rounds. Now the adversary makes its first choice: it chooses canonical teams to which  $p_1$ and $p_2$ belong. These are teams $\{p_1,q_1\}$ and $\{p_2,q_2\}$, where $q_i$ denote the minus agents in each team. (Note that we do not know which of the agents has an even label and which has an odd label in each team). It remains to choose the bases. As $u$, the adversary chooses any node in the line, such that port $1$ is in the right direction, and chooses as $v$ the (unique) node at distance $D$ right  of $u$. Finally the adversary chooses $u$ as the base of team $\{p_1,q_1\}$  and chooses $v$ as 
	the base of team $\{p_2,q_2\}$.  
	
	Case 2. $i < j$.
	
	Now the decisions of the adversary are symmetric with respect to Case 1.
	This time, consider the minus-agents of the teams assigned to $\Sigma$. There are at least $g$ of them.
	By Lemma \ref{dfkp}, there exist two of them, $q_1$ and $q_2$,  such that  if they are placed at any two nodes at distance $D$ and follow their trajectories, then they cannot meet before $\frac{1}{6}D\log g$ rounds. The adversary chooses canonical teams to which  $q_1$ and $q_2$ belong. These are teams $\{p_1,q_1\}$ and $\{p_2,q_2\}$, where $p_i$ denote the plus agents in each team.
	It remains to choose the bases. As $u$ the adversary chooses any node in the line such that port $-1$ is in the right direction, and chooses as $v$ the (unique) node at distance $D$ right  of $u$. Finally the adversary chooses $u$ as the base of team $\{p_1,q_1\}$  and chooses $v$ as the base of team $\{p_2,q_2\}$ (here nothing changes). 
	
	\begin{claim}\label{claim:lb}
		There exists a positive constant $c$ such that, for the above choices of the adversary, the first meeting between agents of different teams occurs after at least $cD\log L$ rounds.
	\end{claim} 
	
	In order to prove the claim, we consider Case 1. The argument in Case 2 is similar.
	The choice of teams $\{p_1, q_1\}$ and $\{p_2, q_2\}$ guarantees that the meeting between agents $p_1$ and $p_2$ requires at least $\frac{1}{6}D\log g= \frac{1}{6}D\log \frac{p}{2}$ rounds, as $g=p/2$.
	
	Agent $q_1$, based at $u$, is a minus-agent and its trajectory is left-progressing; agent $q_2$, based at $v$, is a minus-agent and its trajectory is right-progressing.
	Hence, agents $q_1$ and $q_2$ are diverging.
	As $D> 6H$, agents $q_1$ and $q_2$ can never meet, in view of Lemma \ref{lem-diverging}.
	
	Observe that agent $p_2$ follows agent $q_1$.
	First, suppose that $i>j$.
	The speed $V(q_1)$ of the trajectory of $q_1$ is larger than the speed $V(p_2)$ of the trajectory of $p_2$.
	As $D>2\cdot(72H^2+6H)$, agent $p_2$ can never meet agent $q_1$, in view of Proposition \ref{follow} (i).
	Next, suppose that $i=j$.
	Then $|V(p_2)-V(q_1)|< 1/p$ because $V(q_1)$ and  $V(p_2)$ are both in the interval $I_i$.
	If $V(p_2)-V(q_1)\le 0$, then agent $p_2$ can never meet agent $q_1$, in view of Proposition \ref{follow} (i).
	Otherwise, as $D>2\cdot(72H^2+6H)$, Proposition \ref{follow} (ii) implies that the meeting between $p_2$ and $q_1$
	cannot be achieved before round $3H+36H^2+\lceil (D-6H-72H^2)p\rceil>\frac{D}{2}p \in \omega(D\log p)$.
	Symmetrically, agent $p_1$ follows agent $q_2$.
	Using similar arguments, agents  $p_1$ and $q_2$ either never meet or meet after time $\omega(D\log p)$.
	It follows that for a sufficiently small constant $c'$, any meeting between agents $p_1$ and $q_2$ or $p_2$ and $q_1$ cannot happen before round $c'D\log p$. Hence, for a sufficiently small constant $c$, the first meeting between agents of different teams occurs after at least $cD\log L$ rounds. $\diamond$
	
	
	Claim \ref{claim:lb} concludes the proof of the theorem.
\end{proof}

\subsubsection{The algorithm}

We now describe Algorithm {\tt Small Teams Unoriented} that guarantees gathering teams of size 2 in an unoriented line, in time $O(D\log L)$. More precisely, the algorithm accomplishes gathering of arbitrary $R$ teams of size 2 in time $O(D\log L)$, where $D$ is the distance between the bases of the most distant teams, and the port labeling of the line is arbitrary.

In our algorithm, we will use the following procedure {\tt Dance} $(string, p)$. Its parameter $string$ is a binary sequence, and its parameter $p$ is one of the possible ports $-1$ or $1$. Intuitively, procedure {\tt Dance} $(string, p)$ is an infinite procedure divided into phases of $k$ rounds each, where $k$ is the length of the binary sequence $string$.
In each phase, the agent first chooses the edge $e$ on which the dance will be executed, and then performs the dance itself on edge $e$.

In the first phase, $e$ corresponds to port $p$, and in each subsequent phase, $e$ is the edge incident to the current node, which is different from the edge on which the dance in the previous phase was executed.   
The agent processes $k$ bits of the sequence $string$ in $k$ consecutive rounds as follows: if the bit is 1 then traverse edge $e$, and if the bit is 0 then stay idle. 
Here is the pseudo-code of this procedure.\\

	\noindent
	{\bf Procedure} {\tt Dance} $((b_1,\dots,b_k), p)$
	
	\noindent
	$FirstPhase:= $ {\bf true};\\
	{\bf Repeat forever}\\
	\hspace*{1cm}{\bf if} $FirstPhase$ {\bf then}\\
	\hspace*{2cm}$e:=$ the edge corresponding to port $p$ at the current node\\
	\hspace*{2cm}$FirstPhase:=$ {\bf false};\\
	\hspace*{1cm}{\bf else}\\
	\hspace*{2cm}$q:=$ the port by which the agent entered the current node most recently;\\
	\hspace*{2cm}$e:=$ the edge corresponding to port $-q$;\\  
	\hspace*{1cm}{\bf for} $i:=1$ {\bf to} $k$ {\bf do}\\
	\hspace*{2cm}{\bf if} $b_i=1$ {\bf then}\\
	\hspace*{3cm}traverse $e$\\
	\hspace*{2cm}{\bf else}\\
	\hspace*{3cm}stay idle;\\

%

While procedure {\tt Dance} is formulated as an infinite procedure, it will be interrupted in some round of the execution of the algorithm, and a new procedure {\tt Dance}  with different parameters will be started.

\noindent
{\bf Label transformation.}
Recall that labels of all agents are different integers from the set $\{1,\dots, L\}$.
Let $\len$ denote $\lfloor \log L \rfloor+1$.
Consider a label $\ell \in \{1,\dots, L\}$. We represent it by the unique binary sequence $Bin(\ell)=(b_1b_2\cdots b_{\len})$ which is the binary representation of $\ell$, with a string of zeroes possibly added as a prefix, to get length $\len$.

Let $z=2\cdot\len+4$. The following label transformation is inspired by \cite{DFKP}.
We define the transformed label obtained from $\ell$. This is the binary sequence $Tr(\ell)$ of length $z$ defined as follows:
In $Bin(\ell)=(b_1b_2\cdots b_{\len})$, replace each bit 1 by the string 11, replace each bit 0 by the string 00, add the string 10 as a prefix and the string 00 as a suffix.
Due to the added prefix string 10, $Tr(\ell)$ always contains an odd number of bits 1.
For example, if $L$ is $15$ and $\ell$ is $3$, then $Bin(\ell)$ is the binary sequence $(0011)$, and $Tr(\ell)$ is the binary sequence $(100000111100)$.

\noindent
{\bf Modes.} In each round of the algorithm execution, each agent is in some {\em mode}, encoded in the state of the agent. In each mode, an agent performs some action. Four of these modes instruct the agent to execute procedure {\tt Dance} with various parameters, and three other modes instruct it to stay put. Modes are changed at meetings of the agents, called events.
In Table \ref{tab-modes-unoriented-simple}, we list seven modes in which an agent can be.

\begin{table}[h]
	\centering
	\caption{\label{tab-modes-unoriented-simple}The seven modes in which an agent with label $\ell$ can be, where  $\alpha=Tr(\ell)$ and $\beta=(1110)$. }
	\begin{tabular}{ |c| c| } 
		\hline
		Mode & Action		\\
		\hline
		\hline
		$\negSlow$ & {\tt Dance}		$(\alpha, -1)$\\
		\hline
		$\plusSlow$ & {\tt Dance}		$(\alpha, +1)$\\
		\hline
		$\negFast$ & {\tt Dance}		$(\beta, -1)$\\
		\hline
		$\plusFast$ & {\tt Dance}		$(\beta, +1)$\\
		\hline
		$\idle$ & 	Stay put\\
		\hline
		$\hiber$ & 	Stay put\\
		\hline
		$\senti$ & 	Stay put\\
		\hline
	\end{tabular}
\end{table}

We now explain the roles that are played by the seven modes. 
An agent in any of the modes $\idle$, $\hiber$, or $\senti$ stays at the current node until a new event happens at this node.
Mode $\hiber$ is used only in the first $z$ rounds of the algorithm, while modes $\idle$ and $\senti$ are used in the remaining rounds.
In particular, mode $\senti$ is only assigned to the leftmost and rightmost agents, once these agents identify themselves as such.

Modes $\negSlow$ and $\plusSlow$ will be called {\em slow} modes, and modes $\negFast$ and $\plusFast$ will be called {\em fast} modes.
For simplicity, we call an agent in slow (resp. fast) mode a slow (resp. fast) agent.
In view of procedure {\tt Dance}, a slow or fast agent dances along the same edge in each phase.
We call the edge endpoint, at which a phase starts, the {\em dancing source} of the agent, in this phase.
Observe that the sequences $\alpha$ and $\beta$ used in modes slow and fast have an odd number of bits $1$.
Hence, after each phase, the dancing source of a slow or fast agent is changed to be the other endpoint of this edge.
During an execution of procedure {\tt Dance}, the way that the dancing source changes looks like an object moving on the unoriented line exactly one step in the same direction in each phase.
We call the direction, in which the dancing source of an agent executing {\tt Dance} moves, the {\em progress direction} of the agent.
Although the port numbers at each node are assigned arbitrarily by the adversary and a slow or fast agent cannot tell which direction is left or right, it is always aware of its progress direction.
Indeed, consider an agent $g$ that arrives at node $u$ in some round of the execution of {\tt Dance}. The agent knows if, in this phase,  $u$ is the dancing source or not.
If $u$ is the dancing source of $g$, then the port via which $g$ arrives at $u$, indicates its progress direction;
otherwise, its progress direction is indicated by the other port at $u$.
This is true, because sequences $\alpha$ and $\beta$ start with bit $1$, end with bit $0$, and contain an odd number of bits $1$. 

An agent in mode $\negSlow$ or mode $\negFast$  (resp. mode $\plusSlow$ or mode $\plusFast$) leaves the current node via port $-1$ (resp. $1$), in the first round of the execution of procedure {\tt Dance}.
This is because the second parameter in the procedure is $-1$ (resp. $1$) and sequences $\alpha$ and $\beta$ both start with bit $1$.
We will prove that two slow agents, moving towards each other while executing procedure {\tt Dance}, will meet in a phase, when their dancing sources are at a distance either $1$ or $2$.

Meetings of a fast and a slow agent happen for a different reason.
In view of definitions of sequences $\alpha$ and $\beta$, each phase of procedure {\tt Dance} in a slow mode lasts $z>4$ rounds, while each phase in a fast mode lasts only $4$ rounds.
This means that the dancing source of a slow agent changes every $z$ rounds, while the dancing source of a fast agent changes every four rounds.
Observe that any $4$-bit sub-segment of $\alpha$ is different from $\beta$.

When a fast agent meets a slow agent progressing in the same direction, we will say that the fast agent {\em catches} the slow agent.
Due to the difference of the dancing source changing frequency for slow and fast agents we will show that a fast agent following a slow agent progressing in the same direction always catches it.

The ideas of the label transformation and the procedure \texttt{Dance} are originally from \cite[Procedure \texttt{Extend-Labels}]{DFKP}. The goal of the procedure \texttt{Extend-Labels} is to let two agents meet on graph $K_2$, that is, a graph consisting of two nodes joined by an edge.
During the execution of \texttt{Extend-Labels}, the dancing source of each agent does not change.
We generalize the technique from \cite{DFKP} in two different ways. By performing the procedure \text{Dance},
\begin{itemize}
	\item two slow agents with different labels, moving towards each other and starting with an arbitrary delay at two nodes of a line at an arbitrary distance, meet at the same node,
	\item a slow agent is caught by a fast agent if the former is ahead of the latter and both agents progress in the same direction.
\end{itemize}
As opposed to \cite[Procedure \texttt{Extend-Labels}]{DFKP}, our procedure \texttt{Dance} works in phases, and the dancing source changes in each phase.
Moreover, to guarantee that any $4$-bit sub-segment of $\alpha$ (i.e., a binary sequence used by a slow agent) is different from $\beta$ (i.e., ``1110''), we append a suffix ``00'' in the label transformation for every slow agent.

\noindent
{\bf The high-level idea of the algorithm.}
In the wake-up round, agents are at their bases in teams of size 2. Since they can see each other's labels, they can compare them and assign modes as follows: the agent with smaller label assigns itself mode $\negSlow$ and the agent with larger label assigns itself mode
$\plusSlow$. Conceptually, we number agents $g_1,\dots, g_{2R}$ as follows (see Fig. \ref{fig:2R-agents}). 

\begin{figure}[t]
	\centering
	\includegraphics[width=\textwidth]{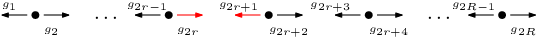}
	\caption{Agents $g_1,\dots, g_{2R}$ on the unoriented line. Agents in red indicate a pair of neighbors.}
	\label{fig:2R-agents}
\end{figure}

Consider the $r$-th team from the left, where $r=1,\dots, R$.
Agent $g_{2r-1}$ is the agent from the $r$-th team that first moves left, and agent $g_{2r}$ is the agent from the $r$-th team that first moves right. Agent $g_{2r-1}$ will be called the left agent of the $r$-th team, and agent $g_{2r}$ will be called the right agent of the $r$-th team.  Notice that an agent does not know in which team it is, and it does not know if it is the left or the right agent in a team,
due to the adversarial  port labeling of the undirected line. Let's call agents $g_{2r}$ and $g_{2r+1}$ {\em neighbors}, for any $r=1,\dots, R-1$.

Our algorithm guarantees that in some round of its execution, all neighbors will meet in pairs.
After such a meeting, both agents switch to mode $\plusFast$, however this change might not happen immediately at the meeting.
If the bases of the $r$-th team and the $(r+1)$-th team are at a distance exactly one, then both neighbors from these teams switch to mode $\hiber$ at their meeting and then switch to mode $\plusFast$ in round $z$.
If the distance between their bases is exactly two, then both neighbors maintain their slow modes until round $z$ and switch to mode $\plusFast$ in this round.
(Note that, in both above cases, the earliest round when both neighbors meet could be earlier than round  $z$). 

Otherwise, when the distance between bases is more than two, then both neighbors switch to mode $\plusFast$ immediately after the meeting.
After changing to a fast mode, an agent swings (possibly switching from mode $\plusFast$ to $\negFast$ or vice versa) between agent $g_1$ and agent $g_{2R}$  both of which are still in slow modes.
Therefore, it is crucial to identify these agents. 
To this end, each agent $e$ in a slow mode keeps a {\em bag} to which it adds labels of fast agents that have caught it.
We will show that only bags of $g_1$ and $g_{2R}$ can have size $(2R-2)$ and that,  in some round,  they will have this size. This is how the pair of agents $g_1$ and $g_{2R}$ is identified. At this time these agents switch mode to  $\senti$.

However, an agent in mode $\senti$ does not know whether it is $g_1$ or $g_{2R}$. We will distinguish $g_1$ from $g_{2R}$ indirectly, using fast agents. The $2R-2$ fast agents travel from one agent in mode $\senti$ to the other, gather their labels, compare them, go to the agent with the smaller label and change state to $\idle$. Without loss of generality, suppose that the label of agent $g_1$ is smaller than that of $g_{2R}$.
As a result, all $2R-2$ fast agents meet $g_1$ and stay with it in mode $\idle$; in other words, $2R-1$ agents gather at the same node in some round.
When this happens, each of these $2R-1$ agents switches to a fast mode and progresses in the direction opposite to their last move.
Hence they will meet agent $g_{2R}$ which stays put in mode $\senti$.
In the end, all $2R$ agents gather at the same node where $g_{2R}$ stays, and transit to state STOP.

We now proceed to the detailed description of Algorithm {\tt Small Teams Unoriented}.

\noindent
{\bf Detailed description of the algorithm.}
As explained in the high-level idea, in each round of the algorithm, each agent is in some mode, and modes may change at {\em events} 
that happen when some agents meet at some node $u$ in a round $t$.
There are four types of events in the algorithm.
Among them, {\tt Event A} is the wake-up event that happens in round 0; {\tt Event B}, {\tt Event C} and {\tt Event D} are meetings before round $z$, in round $z$, and after round $z$, respectively.
Obviously, all four events are pairwise exclusive, hence the choice of the action is unambiguous.
In the sequel, we talk about the actions corresponding to these events.
Essentially, we describe precisely, how modes of agents change at each type of events that happens at a node $u$ in a round $t$. 

For an agent $g$, we define the following variables:
\begin{enumerate}
	\item[] $g.bag$ is used for a slow agent $g$ to store the labels of all fast agents that have caught it;
	\item[] $g.compLabel$ is used for a fast agent to store the label of the recently caught slow agent or the label of the recently met {\tt sentinel};
	\item[] $g.countSenti$ is used for a fast agent to store the number of times it met a {\tt sentinel}.
\end{enumerate}




%
\noindent
{\bf Event A:} All $R$ teams are woken up by the adversary (in round 0). \\

\noindent
$neg := $ the agent with smaller label in the team; \\
$pos := $ the agent with larger label in the team; \\
$neg.bag:=\emptyset$;\\
$pos.bag:=\emptyset$;\\
$neg$ assigns itself mode $\negSlow$;\\
$pos$ assigns itself mode $\plusSlow$; \\

\noindent
{\bf Event B:} The set $\Sigma$ of agents at node $u$ has size larger than 1 in a  round $0<t<z$.\\

\noindent
{\bf for} each agent $g\in \Sigma$ \\
\hspace*{1cm}{\bf if} $g$ is slow and there is another slow agent $g'\in \Sigma$ such that the bases of $g$ and $g'$ are at distance $1$, and $g$ and $g'$ have different progress directions {\bf then}\\
\hspace*{2cm} $g$ switches to mode $\hiber$; \\

\noindent
{\bf Event C:} The set $\Sigma$ of agents at node $u$ has size larger than 1 in round $z$ \\

\noindent
{\bf for} each agent $g\in \Sigma$ \\
\hspace*{1cm} $g.compLabel:=$ empty string; \\
\hspace*{1cm} $g.countSenti:=0$; \\
\hspace*{1cm} $g$ switches to mode ${\plusFast}$; \\

%

	\noindent
	{\bf Event D}: The set $\Sigma$ of agents at node $u$ has size larger than 1 in a round $t>z$.
	
	\noindent
	{\bf case 1:} there are two slow agents $e$ and $f$,  and a set $F$ (possibly empty) of fast agents, in $\Sigma$ \\
	\hspace*{1cm} {\bf for} each agent $g \in \{e \cup f\}$ \\
	\hspace*{2cm} $g.compLabel:=$ empty string; \\
	\hspace*{2cm} $g.countSenti:=0$; \\
	\hspace*{2cm} $g$ switches to mode $\plusFast$;\\
	
	\noindent
	{\bf case 2:} there is a single slow agent $e$ and a (non-empty) set $F$ of fast agents, in $\Sigma$ \\
	\hspace*{1cm} {\bf if} there is no agent in $F$ whose progress direction differs from that of $e$ {\bf then}\\
	\hspace*{2cm} $dir:=$ the progress direction of $e$ in the current round (** $dir\in \{-1,+1\}$ **);\\
	\hspace*{2cm} for each agent $f\in F$\\
	\hspace*{3cm} $e.bag:= e.bag \cup \{f.label\}$;\\
	\hspace*{3cm} $f.compLabel:=e.label$;\\
	\hspace*{3cm} $f$ switches to mode $(-dir)$-{\tt fast};\\
	\hspace*{2cm} {\bf if} the size of $e.bag$ is at least $2R-2$ {\bf then}\\
	\hspace*{3cm} $e$ switches to mode $\senti$; \\

	\noindent
	{\bf case 3:} $\Sigma$ contains only fast agents\\
	\hspace*{1cm} each agent $f\in \Sigma$ maintains its fast mode;\\
	
	\noindent
	{\bf case 4:} $\Sigma$ contains a single sentinel $e$, a non-empty set $F$ of fast agents and a (possibly empty) set $I$ of agents in mode $\idle$ \\
	\hspace*{1cm} $dir := $ the common progress direction of agents in $F$ in the current round; \\
	\hspace*{1cm} {\bf if} the size of $\Sigma$ is $2R-1$ {\bf then} \\
	\hspace*{2cm} $e$ switches to mode $(-dir)$-{\tt fast}; \\
	\hspace*{2cm} each agent in $F\cup I$ switches to mode $(-dir)$-{\tt fast}; \\
	\hspace*{1cm} {\bf else}\\
	\hspace*{2cm} {\bf if} the size of $\Sigma$ is $2R$ {\bf then} \\
	\hspace*{3cm} each agent at $u$ transits to state {\tt STOP}; \\
	\hspace*{2cm} {\bf else} \\ 
	\hspace*{3cm} {\bf for} each agent $f\in F$\\
	\hspace*{4cm} $f.countSenti:=f.countSenti+1$;\\
	\hspace*{4cm} {\bf if} $f.countSenti < 2$  {\bf then}\\
	\hspace*{5cm} $f.compLabel:=e.label$;\\
	\hspace*{5cm} $f$ switches to mode $(-dir)$-{\tt fast}; \\
	\hspace*{4cm} {\bf else}\\
	\hspace*{5cm} {\bf if} $f.compLabel>e.label$ {\bf then}\\
	\hspace*{6cm} $f$ switches to mode $\idle$; \\
	\hspace*{5cm} {\bf else}\\
	\hspace*{6cm} $f.compLabel:=e.label$;\\
	\hspace*{6cm} $f$ switches to mode $(-dir)$-{\tt fast};\\

	\noindent
	{\bf case 5:} $\Sigma$ contains a single {\tt sentinel} $e$ and a non-empty set $I$ of agents in mode $\idle$\\
	\hspace*{1cm} $e$ maintains its mode $\senti$; \\ 
	\hspace*{1cm} each agent in $I$ maintains its mode $\idle$;\\


This completes the detailed description of Algorithm {\tt Small Teams Unoriented}. A terminological precision is in order. When we say that an agent switches to some mode, it means that it starts executing procedure {\tt Dance} with appropriate parameters, from its beginning. 
When we say that an agent maintains some mode, it means that the agent continues to execute procedure {\tt Dance} by processing the next bit of the same string without interruption.

We will have to prove that, for Event D, the five cases enumerated in the description are the only ones that can happen during a meeting after round $z$, and that only the left-most and right-most agents become sentinels.
Hence, in case 4,
all fast agents have a common progress direction, as they meet a sentinel, and hence their progress direction before the meeting had to be towards this sentinel.

\subsubsection{Correctness and complexity}

In this section, we present the proof of correctness of Algorithm {\tt Small Teams Unoriented}, and establish its complexity.

\begin{theorem}\label{unoriented-2}
	Suppose that agents have distinct labels drawn from the set $\{1, \cdots, L\}$, that all teams are of size $2$, and that the distance between the bases of the most distant teams is $D$.
	Algorithm {\tt Small Teams Unoriented} gathers all agents at the same node of an unoriented line in time $O(D\log L)$.
\end{theorem}

The high-level idea of the proof is the following.
Without loss of generality, assume that the label of $g_1$ is smaller than the label of $g_{2R}$.
First, we prove that agent $g_{2r}$ meets its neighbor $g_{2r+1}$, for each $1\le r< R$.
After the meeting, $g_{2r}$ and $g_{2r+1}$ become fast agents and swing between $g_1$ and $g_{2R}$.
Then, we show that only $g_1$ and $g_{2R}$ will change (possibly in different rounds) to mode $\senti$.
As agents $g_2, \dots, g_{2R-1}$ swing between $g_1$ and $g_{2R}$, and $g_1$ is in mode $\senti$, agents $g_2, \dots, g_{2R-1}$ will meet $g_1$ and stay with it in mode $\idle$.
Next, we prove that by the earliest round when all $2R-2$ fast agents gather at the same node where $g_1$ stays, agent $g_{2R}$ has already changed to mode $\senti$.
As a result, all $2R-1$ agents, apart from $g_{2R}$, become fast agents, progress towards $g_{2R}$ and eventually meet $g_{2R}$, while $g_{2R}$ stays put in mode $\senti$. Then gathering is achieved and all agents transit to state STOP.
More formally, let $t_1$ denote the earliest round in which agent $g_i$ becomes a fast agent after meeting its neighbor, for all $1<i<2R$.
Let $t_2$ denote the earliest round in which both $g_1$ and $g_{2R}$ switched to mode $\senti$.
Let $t_3$ denote the earliest round in which all $2R-1$ agents, apart from $g_{2R}$, have gathered at the same node.
Let $t_4$ denote the earliest round in which all $2R$ agents gather at the same node.
Our goal is to prove that all rounds $t_1, \dots, t_4$ exist; furthermore, we will show that $t_1<t_2\leq  t_3<t_4\in O(D\log L)$.

The rest of the section is devoted to the detailed proof of Theorem \ref{unoriented-2}. The proof is split into several lemmas.
The first lemma shows by which round any pair of neighbors change to mode $\plusFast$.
\begin{lemma}
	\label{lem-t1}
	Consider any pair of neighbors $g_{2r}$ and $g_{2r+1}$ and let $\mathcal{D}_r$ denote the distance between their bases, where $1\le r< R$.
	Agents $g_{2r}$ and $g_{2r+1}$ switch to mode $\plusFast$ by round $\lceil \mathcal{D}_r /2\rceil z$.
\end{lemma}

\begin{proof}
	We start with the following general claim.
	
	\begin{claim}
		\label{claim-meeting-slow-agents}
		Consider any pair of slow agents $e$ and $f$, starting at two different nodes at a distance $d$, such that the progress direction of each agent is towards the base of the other. Then both agents meet in a phase of procedure {\tt Dance}
		when their dancing sources at a distance either $1$ or $2$. This takes less than $\lceil d/2 \rceil z$ rounds.
	\end{claim}
	
	In order to prove the claim, 
	let $\ell_e$ and $\ell_f$ denote the labels of agents $e$ and $f$, respectively.
	Since both agents are woken up simultaneously and are in slow modes, their dancing sources are always changed in the same round.
	Initially, the distance of their dancing sources is $d$.
	After each $z$ rounds, the distance between their dancing sources is decreased by $2$.
	So, during the execution of {\tt Dance}, their dancing sources' distance changes to be $d, d-2, d-4, \dots$.
	There are two cases to consider: either $d$ is even or $d$ is odd.
	First, suppose that $d$ is even.
	In the first round of the phase when their dancing sources' distance becomes $2$, both agents leave their respective dancing sources and meet at the node between their dancing sources, since both $Tr(\ell_e)$ and $Tr(\ell_f)$ start with bit $1$.
	Therefore, if $d$ is even, then both agents meet in the first round of the phase when the dancing sources' distance is $2$, taking $(d-2)/2\cdot z+1< d\cdot z/2$ rounds.
	Next, suppose that $d$ is odd.
	Observe that there is at least one position at which the bit in $Tr(\ell_e)$ differs from the bit in $Tr(\ell_f)$.
	Let $t$ denote the first round of the phase when the dancing sources' distance becomes one and let $p$ denote the leftmost position at which the bits in $Tr(\ell_e)$ and $Tr(\ell_f)$ differ.
	Notice that $2<p<z-1$.
	Then $e$ and $f$ meet in round $t+p-1$, because in this round one agent traverses the edge of which the endpoints are the dancing sources of $e$ and $f$ and the other agent stays put in this round.
	Therefore, if $d$ is odd, then both agents meet when the dancing sources' distance is $1$, taking $t+p-1$ rounds, where $t+p-1<(d-1)z/2+z\le \lceil d/2\rceil z$, as $t=(d-1)z/2$. This completes the proof of the claim. $\diamond$

	In the first $z$ rounds, agents $g_2, g_4, \dots, g_{2R}$ dance between their respective bases and the right adjacent nodes, while agents $g_{1}, g_{3}, \dots, g_{2R-1}$ dance between their respective bases and the left adajcent nodes.
	Therefore, only neighbors whose bases are at a distance at most $2$ can meet within the first $z$ rounds. In the next two claims we show that neighbors whose bases are at distance at most 2 switch to mode $\plusFast$ in round $z$. First suppose that $\mathcal{D}_r=1$.
	
	\begin{claim}
		\label{claim-base-dist-1}
		Consider any pair of neighbors $g_{2r}$ and $g_{2r+1}$.
		If $\mathcal{D}_r=1$, then $g_{2r}$ and $g_{2r+1}$ are at the same node and change to mode $\plusFast$ in round $z$.
	\end{claim}
	
	For simplicity, we call $g_{2r}$ and $g_{2r+1}$, {\em close neighbors}, if $\mathcal{D}_r=1$.
	Consider agent $g_{2r}$ at some node of the line.
	Let $g_{j}$ denote any agent such that $g_{j}$ and $g_{2r}$ have different progress directions and the bases of both agents are at a distance $1$.
	If $g_j$ exists, then $j$ can be either $2r+1$ or $2r-3$.
	Notice that $g_{2r}$ and $g_{2r-3}$, as slow agents, cannot meet within the first $z$ rounds, while $g_{2r+1}$, as the close neighbor of $g_{2r}$, meets $g_{2r}$ within the first $z-1$ rounds, in view of Claim \ref{claim-meeting-slow-agents}.
	Consider agent $g_{2r+1}$ at some node of the line.
	Let $g_{j'}$ denote any agent such that $g_{j'}$ and $g_{2r+1}$ have different progress directions and the bases of both agents are at a distance $1$.
	If $g_{j'}$ exists, then $j'$ can be either ${2r}$ or $2r+4$.
	Notice that $g_{2r+1}$ and $g_{2r+4}$, as slow agents, cannot meet within the first $z$ rounds, while $g_{2r}$, as the close neighbor of $g_{2r+1}$, meets $g_{2r+1}$ within the first $z-1$ rounds, in view of Claim \ref{claim-meeting-slow-agents}.
	
	If there is a pair of close neighbors, $g_{2r}$ and $g_{2r+1}$, then they meet at some node $u$ in some round $t$, where $t<z$.
	At their meeting, {\tt Event B} is triggered.
	Although there might be more than two agents at node $u$ in round $t$, the above argument implies that each agent at $u$ can be aware of whether its close neighbor is present or not by checking the distances of its base and other agents' bases and comparing its progress direction with other agents' progress directions.
	In view of the action in {\tt Event B}, an agent switches to mode $\hiber$ only if its close neighbor is also present at $u$ in round $t$.
	After the action, all pairs of close neighbors at $u$ switch to mode $\hiber$.
	So, an agent in mode $\hiber$ always stays at the same node with its close neighbor, hence each agent in mode $\hiber$ switches in round $z$ to mode $\plusFast$, in view of the action in {\tt Event C}.
	Therefore, if $\mathcal{D}_r=1$, then both neighbors $g_{2r}$ and $g_{2r+1}$ are at the same node and switch to mode $\plusFast$ in round $z$. 
	This proves the claim. $\diamond$
	
	Next, suppose that $\mathcal{D}_r=2$.
	\begin{claim}
		\label{claim-base-dist-two}
		Consider any pair of neighbors $g_{2r}$ and $g_{2r+1}$.
		If $\mathcal{D}_r=2$, then $g_{2r}$ and $g_{2r+1}$ are at the same node and change to mode $\plusFast$ in round $z$.
	\end{claim} 
	
	In view of Claim \ref{claim-meeting-slow-agents}, $g_{2r}$ and $g_{2r+1}$ meet in the first round.
	This meeting triggers {\tt Event B}, but both agents cannot switch to mode $\hiber$ immediately at the meeting, since they are not close neighbors.
	In round $z$, both agents meet again at the unique node between their bases, and this time {\tt Event C} is triggered.
	So both agents switch to mode $\plusFast$ in round $z$.
	This proves the claim. $\diamond$
	
	Notice that if an agent has not met its neighbor within the first $z$ rounds, then it cannot meet other agents in round $z$.
	Indeed, consider any agent $g_{2r+1}$ that has not met its neighbor $g_{2r}$.
	The distance between the bases of agents $g_{2r+1}$ and $g_{2r}$ has to be at least 3; otherwise, agents $g_{2r+1}$ and $g_{2r}$ would meet in round $z$, in view of Claims \ref{claim-base-dist-1} and \ref{claim-base-dist-two}.
	Therefore, none of the agents $g_{1}, \dots, g_{2r}$ has met $g_{2r+1}$ by round $z$.
	On the other hand, none of the agents $g_{2r+2}, \dots, g_{2R}$ can reach any node that is left of the base of $g_{2r+2}$, while agent $g_{2r+1}$ stays at the left adjacent node of its base in round $z$.
	Therefore, agent $g_{2r+1}$ is alone in round $z$ and {\tt Event C} cannot apply to agent $g_{2r+1}$.
	The argument for any agent $g_{2r}$ that has not met its neighbor $g_{2r+1}$ is similar.
	
	Now, we consider meetings after round $z$. In the next two claims we show that in any round $t>z$, there are at most two slow agents present at each node.
	
	\begin{claim}
		\label{claim-left-left-right-right}
		Let $g_i$ and $g_j$ denote any two slow agents in any round $t>z$ such that $g_i$ and $g_j$ share the same progress direction.
		Then
		\begin{enumerate}[label=(\roman*)]
			\item if $i<j$, then $g_i$ is left of $g_j$ in round $t$, and
			\item if $i>j$, then $g_i$ is right of $g_j$ in round $t$.
		\end{enumerate}
	\end{claim}
	
	In order to prove the claim, 
	observe that no agents switch to either mode $\negSlow$ or mode $\plusSlow$, after round 1.
	So agents $g_i$ and $g_j$ have been slow agents between round 1 and round $t$.
	Since neither $g_i$ nor $g_j$ becomes a fast agent in round $z$, then the bases of $e$ and $f$ are at a distance at least $3$, in view of Claims \ref{claim-base-dist-1} and \ref{claim-base-dist-two}.
	If $i<j$, then the base of $g_i$ is left of the base of $g_j$; if $i>j$, then base of $g_i$ is right of the base of $g_j$.
	The distance of their dancing sources remains the same in the first $t$ rounds: this distance is always at least $3$. 
	So, the relative positions between $g_i$ and $g_j$ are always the same in the first $t$ rounds. $\diamond$
	
	\begin{claim}
		\label{claim-at-most-2-slow}
		After round $z$, there can be at most two slow agents at the same node.
		Furthermore, if there are two slow agents at the same node, then their progress directions have to be different.
	\end{claim}
	
	In view of Claim \ref{claim-left-left-right-right}, any two slow agents that share the progress direction cannot meet after round $z$.
	Notice that if at least three slow agents were at the same node, at least two of them would share the same progress direction.
	Therefore, there cannot be more than two slow agents at the same node after round $z$.
	If there are two slow agents at the same node, then they have different progress directions.
	The claim is proved. $\diamond$
	
	In the next two claims we prove that none of the agents $g_2,\dots, g_{2R-1}$ ever switches to mode $\senti$.
	
	\begin{claim}
		\label{claim-g1-g2-slow-senti}
		After round $z$, neither $g_1$ nor $g_{2R}$ meets a slow agent; $g_1$ and $g_{2R}$ remain slow until they switch to mode $\senti$.
	\end{claim}
	
	We give the proof for $g_1$, while the proof for $g_{2R}$ is similar.
	After round $z$, $g_1$ is still a slow agent.
	Indeed, within the first $z-1$ rounds, the only agents that $g_1$ can meet are $g_2$ and $g_3$. If $g_1$ meets $g_2$, then {\tt Event B} is triggered, but $g_1$ cannot switch to mode $\hiber$, as $g_1$ and $g_2$ share the same base; if $g_1$ meets $g_3$, then {\tt Event B} is triggered as well, but $g_1$ cannot switch to mode $\hiber$, as $g_1$ and $g_3$ share the progress direction.
	In round $z$, agent $g_1$ stays alone at the node that is left of its base, so $g_1$ cannot switch to mode $\plusFast$.
	
	After round $z$, $g_1$ cannot meet any slow agent that has the same progress direction, in view of Claim \ref{claim-left-left-right-right}.
	On the other hand, any slow agent that has a different progress direction than $g_1$ cannot meet $g_1$ after round $z$, due to the fact that $g_1$ is the leftmost agent on the line and it progresses left.
	Therefore, $g_1$ cannot meet any slow agent, after round $z$.
	After round $z$, a slow agent switches to either mode $\plusFast$ or mode $\senti$.
	A slow agent switches to mode $\plusFast$ only if it meets another slow agent, in view of the action in {\tt Event D} (more precisely {\tt case 1}).
	So agent $g_1$ remain slow until it switches to mode $\senti$. $\diamond$
	
	\begin{claim}
		\label{claim-no-senti-2-2R-1}
		None of the agents $g_2,\dots, g_{2R-1}$ ever switches to mode $\senti$.
	\end{claim}
	
	In view of the action in {\tt Event D}, more precisely {\tt case 2}, only a slow agent can switch to mode $\senti$, after it has seen $2R-2$ fast agents, and a $\senti$ only appears after round $z$.
	Consider a slow agent $g_i$ for any $1<i<2R$.
	In view of Claim \ref{claim-g1-g2-slow-senti}, $g_1$ and $g_{2R}$ are slow agents in the first $z$ rounds, and $g_1$ and $g_{2R}$ cannot meet any slow agent after round $z$.
	So $g_i$ would not add their labels to its bag, even if it met $g_1$ or $g_{2R}$ in the first $z$ rounds.
	As a slow agent, $g_i$ can only meet at most $2R-3$ fast agents.
	So it cannot switch to mode $\senti$. $\diamond$
	
	We are now ready to consider meetings of neighbors with bases at distance more than 2.
	
	\begin{claim}
		\label{claim-base-dist-3-or-more}
		Consider any pair of neighbors $g_{2r}$ and $g_{2r+1}$.
		If $\mathcal{D}_r>2$, then $g_{2r}$ and $g_{2r+1}$ are at the same node and switch to mode $\plusFast$ before round $\lceil \mathcal{D}_r/2\rceil z$.
	\end{claim}
	
	The remark following the proof of Claim \ref{claim-base-dist-two} tells that agent $g_{2r}$ stays alone at the 
	right adjacent node of its base in round $z$, and agent $g_{2r+1}$ stays alone at the left adjacent node of its base in round $z$.
	So, their distance in round $z$ is $\mathcal{D}_r-2$. 
	Next, we show that the first meeting that either $g_{2r}$ or $g_{2r+1}$ can have with a slow agent after round $z$ is meeting each other.
	
	In view of the action in {\tt Event D}, a slow agent can switch to mode $\senti$ or mode $\plusFast$ after round $z$.
	In view of Claim \ref{claim-no-senti-2-2R-1}, neither $g_{2r}$ nor $g_{2r+1}$ switches to mode $\senti$.
	Therefore, they can only switch to mode $\plusFast$.
	Agent $g_{2r}$ switches to mode $\plusFast$, after it meets a slow agent, and this slow agent must have a different progress direction, in view of Claim \ref{claim-at-most-2-slow}. 
	
	Consider any slow agent $g_i$.
	If $i$ is even, then $g_i$ cannot meet $g_{2r}$, as both of them share the same progress direction.
	If $i$ is odd and $i<2r$, then $g_i$ as a slow agent cannot reach any node right of the base of $g_{2r}$, after round $z$.
	So, $g_i$ cannot meet $g_{2r}$ after round $z$.
	If $i$ is odd and $i>2r+1$, then $g_{2r}$ would meet $g_{i}$ later than $g_{2r}$ meets $g_{2r+1}$, since $g_{i}$ is right of $g_{2r+1}$, in view of Claim \ref{claim-left-left-right-right}.
	Therefore, after round $z$, $g_{2r}$ meets $g_{2r+1}$ earlier than any other slow agents.
	The claim that $g_{2r+1}$ meets $g_{2r}$ earlier than any other slow agents after round $z$ can be proved in a similar way.
	
	In view of Claim \ref{claim-meeting-slow-agents}, $g_{2r}$ and $g_{2r+1}$ meet before round $\lceil \mathcal{D}_r/2\rceil z$ and they switch to mode $\plusFast$ at the meeting.
	This proves the claim. $\diamond$
	
	Therefore, Claims \ref{claim-base-dist-1}, \ref{claim-base-dist-two} and \ref{claim-base-dist-3-or-more} imply Lemma \ref{lem-t1}.
\end{proof}

In view of Lemma \ref{lem-t1}, we have $t_1 \leq \lceil {D}/2\rceil z$, as ${D}\geq\max_{1\le r <R} \{\mathcal{D}_r\}$.
Next, we consider what happens after round $t_1$.
%
%
The following lemma says that a fast agent always meets a slow agent that it follows.

\begin{lemma}
	\label{lem-catch-slow-fast}
	Consider a slow agent $e$ and a fast agent $f$ sharing the same progress direction in some round $t$.
	Suppose that in this round, $f$ follows $e$ and the distance between $e$ and $f$ is $d$.
	Then, agent $f$ meets $e$ by round $t+48d+7$.
\end{lemma}
\begin{proof}
	Without loss of generality, we assume that the progress direction shared by both $e$ and $f$ is left.
	So, in round $t$, $f$ is right of $e$ (at distance $d$), and the distance $\delta$ between their dancing sources is at most $d+1$.
	Let $\tau$ denote the earliest round in which $f$ is left of $e$.
	Let $u_f$ denote the node reached by $f$ in round $\tau$, and let $u_e$ denote the node reached by $e$ in round $\tau$.
	So, $u_f$ and $u_e$ are the endpoints of a common edge and $u_f$ is left of $u_e$ (see Fig \ref{fig:u-e-f}).
	
	\begin{figure}[t]
		\centering
		\includegraphics[width=\textwidth]{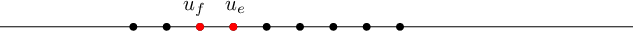}
		\caption{The positions of $u_e$ and $u_f$ on the line.}
		\label{fig:u-e-f}
	\end{figure}
	
	\begin{claim}\label{tau}
		$\tau\le t+48d+4$.
	\end{claim}
	
	In order to prove the claim, notice that as $f$ is a fast agent, its dancing source moves one step left every four rounds; as $e$ is a slow agent, its dancing source moves one step left every $z=2\cdot\len+4$ rounds.
	By round $t+48d+4$, the dancing source of $f$ has moved left at least $12d$ steps, while the dancing source of $e$ can move left at most $6d+1$ steps\footnote{Since $R\ge2$, at least four distinct labels are required and $\len\ge 2$.}, as $\len\ge 2$.
	By round $t+48d+4$, the dancing source of $f$ must be left of the dancing source of $e$ and the distance between these dancing sources is at least 2, since $12d> 6d+1+\delta$, $d\ge 1$ and $\delta\le d+1$.
	Hence, $f$ is left of $e$ by round $t+48d+4$. This proves the claim. $\diamond$
	
	If $e$ and $f$ have met by round $\tau$, then the statement of the lemma follows from Claim \ref{tau}. Hence, we may assume that $e$ and $f$ have not met by round $\tau$.
	
	\begin{claim}\label{moves tau}
		In round $\tau$, $f$ moves from $u_e$ to $u_f$ and $e$ moves from $u_f$ to $u_e$.
	\end{claim}
	
	We prove the claim by contradiction.
	If neither $f$ nor $e$ moves in round $\tau$, then $f$ would be left of $e$ in round $\tau-1$, which is impossible, as $\tau$ is the earliest round in which $f$ is left of $e$.
	Hence, at least one of them is moving in round $\tau$.
	If $f$ stays and $e$ moves in round $\tau$, then $f$ stays at $u_f$ and $e$ moves from $u_f$ to $u_e$.
	However, in this case, agents $e$ and $f$ have met at node $u_f$ in round $\tau-1$, which contradicts our assumption.
	If $f$ moves and $e$ stays in round $\tau$, then $f$ moves from $u_e$ to $u_f$ and $e$ stays at $u_e$.
	In this case, agents $e$ and $f$ have met at $u_e$ in round $\tau-1$, which again contradicts our assumption.
	Therefore, in round $\tau$, $f$ moves from $u_e$ to $u_f$ and $e$ moves from $u_f$ to $u_e$, which proves the claim. $\diamond$

	We will now show that agents $e$ and $f$ have met by round $\tau+3$.
	Let $\ell_e$ denote the label of $e$.
	Recall that the strings processed by $e$ and $f$, while executing procedure {\tt Dance}, are $Tr(\ell_e)$ and $(1110)$, respectively.
	Let $i$ denote the position of the bit in $Tr(\ell_e)$ that agent $e$ is processing in round $\tau$, and let $j$ denote the position of the bit in $(1110)$ that agent $f$ is processing in round $\tau$.
	Since the progress direction of $e$ is left and $e$ moves right from $u_f$ to $u_e$ in round $\tau$, it follows that $u_e$ is the dancing source of $e$ in round $\tau$. Since in round $\tau$, $e$ moves back to its dancing source, $i$ must be an odd number and $1<i\le 2\cdot\len+1$.
	Since the progress direction of $f$ is left and $f$ moves left from $u_e$ to $u_f$ in round $\tau$, it follows that $u_e$ is the dancing source of $f$ in round $\tau$. Since agent $f$ leaves its dancing source in round $\tau$, $j$ must be either $1$ or $3$.
	If $j=1$ then $e$ and $f$ meet by round $\tau+3$.
	This is because any 4-bit sub-segment of $Tr(\ell_e)$ is different from $(1110)$ and there must exist at least one round among rounds $\tau+1$, $\tau+2$ and $\tau+3$, in which either $e$ traverses the edge $(u_f, u_e)$ and $f$ stays put or $f$ traverses the edge and $e$ stays put.
	If $j=3$ then $e$ and $f$ meet in round $\tau+1$.
	This is because in round $\tau+1$, $e$ traverses the edge $(u_f, u_e)$ and $f$ stays put.
	Therefore, by round $\tau+3$, agents $e$ and $f$ have met.
	Since $\tau\le t+48d+4$, agents $e$ and $f$ meet by round $t+48d+7$.
\end{proof}

Next, we prove that both $g_1$ and $g_{2R}$ will switch to mode $\senti$ (possibly in different rounds), which implies the existence of round $t_2$.

\begin{lemma}
	\label{lem-both-switch-to-senti}
	Both $g_1$ and $g_{2R}$ switch (possibly in different rounds) to mode $\senti$.
\end{lemma}

\begin{proof}
	Observe that in the first $t_1$ rounds, neither $g_1$ nor $g_{2R}$ switches to mode $\senti$ and after round $t_1$, none of agents $g_2, \dots, g_{2R-1}$ is a slow agent.
	In view of Lemma \ref{lem-catch-slow-fast}, a fast agent always meets a slow agent, if the former one follows the latter one.
	When agent $g_1$ or $g_{2R}$ in a slow mode is met by a fast agent, then the fast agent only changes its progress direction after the meeting.
	Hence, agents $g_2, g_3, \dots, g_{2R-1}$ in fast modes swing between $g_1$ and $g_{2R}$, until at least one of the agents $g_1$ and $g_{2R}$ adds $2R-2$ labels of fast agents to its bag and then switches to mode $\senti$.

	First, suppose that $g_{2R}$ has switched to mode $\senti$ and $g_{1}$ is still a slow agent.
	Let $f$ denote any fast agent and let $\tau$ denote the earliest round in which $f$ meets $g_{2R}$ and $g_{2R}$ is already in mode $\senti$.
	We show that $f.countSenti$ cannot become 2 in round $\tau$.
	Indeed, if $f.countSenti$ became 2 in round $\tau$, then $f$ would have met a sentinel $s$ before round $\tau$.
	In view of the assumption, $g_{2R}$ is the only sentinel in the first $\tau$ rounds.
	Therefore, $s$ has to be $g_{2R}$.
	Then $\tau$ could not be the earliest round in which $f$ meets $g_{2R}$, and $g_{2R}$ is in mode $\senti$.
	This is a contradiction.
	If $f.countSenti$ becomes 1 in round $\tau$, then $f$ remains fast but changes its progress direction.
	After round $t_1$, none of the agents $g_2, \dots, g_{2R-1}$ is a slow agent, so none of them could cause agent $f$ to change its progress direction.
	Hence, in view of Lemma \ref{lem-catch-slow-fast}, agent $f$ catches $g_1$.
	As $f$ is any fast agent, $g_{1}$ meets all $2R-2$ fast agents sooner or later and switches mode to $\senti$.
	
	Second, suppose that $g_{1}$ switches to mode $\senti$ and $g_{2R}$ never switches to mode $\senti$.
	Hence, there exists at least one fast agent $g_j$, where $1<j<2R$, whose label is never added to the bag of $g_{2R}$.
	This means that before $g_j$ meets $g_{2R}$, it has stopped swinging, switched to mode $\idle$
	and stayed put with $g_1$.
	In view of case 4 in {\tt Event D}, agent $g_j$ only switches to mode $\idle$ after $g_j.countSenti\ge 2$.
	Hence, agent $g_j$ only switches to mode $\idle$ after it met $g_1$ at least twice, since $g_{1}$ is the only agent in mode $\senti$.
	However, after $g_1$ switches to mode $\senti$ and $g_j$ meets $g_1$ for the first time, $g_j$ cannot meet $g_1$ for the second time without meeting $g_{2R}$.
	This is because after the fast agent $g_j$ left $g_1$, the only slow agent that can change the progress direction of $g_j$ is $g_{2R}$.
	Hence $g_j$ must meet $g_{2R}$ and $g_{2R}$ must add the label of $g_j$ to its bag. This is a contradiction.
	This proves that $g_{2R}$ must  eventually switch to mode $\senti$ as well.
\end{proof}

In view of Lemma \ref{lem-both-switch-to-senti}, both $g_1$ and $g_{2R}$ eventually switch to mode $\senti$, so round $t_2$ exists.
Furthermore, since no sentinel appears in the first $t_1$ rounds, we have $t_1<t_2$.
When a fast agent $f$ meets a sentinel and $f.countSenti$ is increased to be at least 2 at the meeting, $f$ has to identify whether this sentinel is $g_1$ or $g_{2R}$ and then take different actions correspondingly.
We will use the following Lemma.

\begin{lemma}
	\label{lem-comp-labels}
	Consider any fast agent $g_i$ which meets the sentinel $g_1$ (resp. $g_{2R}$)  at node $u$ in round $t$.
	Suppose that the number of agents at node $u$ in round $t$ is smaller than  $2R-1$.
	If $g_i.countSenti$ is increased to a value $y\geq 2$ in round $t$, then $g_i.compLabel$ stores the label of $g_{2R}$ (resp. $g_1$) in this round.
\end{lemma}
\begin{proof}
	Let $\tau$ denote the round in which $g_i.countSenti$ becomes $1$.
	In round $\tau$, $g_i$ meets a sentinel $s$.
	At the meeting, $g_i.countSenti$ is increased from $0$ to $1$, $g_i.compLabel$ is set to be the label of $s$, and $g_i$ remains fast but changes its progress direction.
	Observe that after round $\tau$, none of the agents $g_2,\dots, g_{2R-1}$ is a slow agent or a sentinel.
	Hence,  after round $\tau$, only a meeting with agent $g_1$ or $g_{2R}$ can change the value of $g_i.compLabel$.
	If $s$ is $g_1$, then agent $g_i$ meets agent $g_{2R}$ before meeting $g_1$ for the next time.
	At the meeting of $g_i$ and $g_{2R}$, $g_i.compLabel$ stores the label of $g_1$ and is updated to be the label of $g_{2R}$.
	If $s$ is $g_{2R}$
	then agent $g_i$ meets agent $g_{1}$ before meeting $g_{2R}$ for the next time.
	At the meeting of $g_i$ and $g_{1}$, $g_i.compLabel$ stores the label of $g_{2R}$ and is updated to be the label of $g_{1}$.
	
	Let $\tau'$ denote the round in which $g_i$ switches to mode $\idle$.
	Notice that $g_i.compLabel$ remains unchanged after round $\tau'$ and that $g_i$ alternates between meeting $g_1$ and $g_{2R}$, between rounds $\tau+1$ and $\tau'$.
	During this period, $g_i.countSenti\ge 1$, and each time $g_i$ meets $g_1$ (resp. $g_{2R}$), $g_i.compLabel$ stores the label of $g_{2R}$ (resp. $g_{1}$) and is updated to the label of $g_1$ (resp. $g_{2R}$) .
	This proves the lemma.
\end{proof}

Let $s$ denote the sentinel at the meeting described in Lemma \ref{lem-comp-labels}.
At this meeting, a fast agent $f$ already knows the labels of both agents $g_1$ and $g_{2R}$, and hence can identify the one with the smaller label
by comparing $f.compLabel$ with the label of $s$. Recall that we assumed (without loss of generality) that the label of $g_1$ is smaller than the label of $g_{2R}$.
If $s.label<f.compLabel$, then $s$ is $g_1$, and $f$ switches to mode $\idle$ and stays with $g_1$; otherwise $s$ is $g_{2R}$, so $f$ sets $f.compLabel$ to be the label of $g_{2R}$ and remains fast but changes its progress direction.
Therefore, all the fast agents $g_2, \dots, g_{2R-1}$ can identify $g_1$ (which is already in mode {\tt sentinel}) and switch to mode $\idle$ at the meeting.
Eventually, $2R-1$ agents gather at the same node where $g_1$ stays put.
This implies the existence of round $t_3$.
Next, we prove that $t_2\le t_3$.

\begin{lemma}
	\label{lem-t-3-t-2}
	By round $t_3$, $g_{2R}$ has switched to mode $\senti$.
\end{lemma}
\begin{proof} 
	Let $\tau_i$, for $1<i<2R$, denote the earliest round when $g_i$ switches to mode $\idle$.
	
	\begin{claim}
		By round $\tau_i$, $g_{2R}$ has met $g_i$ and added the label of $g_i$ to its bag.
	\end{claim}
	
	In order to prove the claim, notice that both $g_1$ and $g_{2R}$ switch to mode $\senti$, in view of Lemma \ref{lem-both-switch-to-senti}.
	If $g_{2R}$ has already switched to mode $\senti$ by round $\tau_i$, then the claim immediately follows.
	Otherwise, $g_1$ is in mode $\senti$ in round $\tau_i$, while $g_{2R}$ is still a slow agent. 
	Let $t'$ denote the earliest round when $g_1$ switches to mode $\senti$.
	Obviously, $t'< \tau_i$.
	Since $g_i.countSenti\ge 2$ and $g_1$ is the only agent in mode $\senti$ in round $\tau_i$, $g_i$ has met $g_1$ at least twice after round $t'$.
	When $g_i$ meets $g_1$ for the first time after round $t'$, $g_i$ increases $g_i.countSenti$ to $1$, and then it leaves $g_1$ as a fast agent and progresses towards $g_{2R}$.
	At this time, none of the agents $g_2, \dots, g_{2R-1}$ is a slow agent, so $g_{2R}$ is the only agent that can change the progress direction of $g_i$, so that $g_i$ can meet $g_1$ later for the second time after round $t'$.
	Therefore, by round $\tau_i$, $g_{2R}$ has met $g_i$ and added the label of $g_i$ to its bag. This proves the claim. $\diamond$
	
	Recall that $t_3 = \max_{1<i<2R}\{\tau_i\}$.
	In view of the above claim, it follows that $g_{2R}$ has switched to mode $\senti$ by round $t_3$.
\end{proof}

In view of Lemma \ref{lem-t-3-t-2}, by round $t_3$, both agents $g_1$ and $g_2$ have switched to mode $\senti$.
So, we have $t_2\le t_3$.
In round $t_3$, all agents $g_1,\dots,g_{2R-1}$ are at the same node, become fast agents and progress towards $g_{2R}$, while $g_{2R}$ stays put.
Eventually, all $2R$ agents gather at the same node where $g_{2R}$ stays put.

This concludes the proof that gathering is eventually achieved.
In order to complete the proof of correctness of Algorithm {\tt Small Teams Unoriented}, we need to show that the algorithm is correctly formulated, i.e., that the five cases enumerated in {\tt Event D} are the only
ones that can happen during a meeting after round $z$.

An agent after round $z$ can be in one of the following four types of modes: slow, fast, $\senti$ or $\idle$.
It follows from the description of the algorithm that the following statements are all true after round $z$.
\begin{enumerate}[label=\roman*)]
	\item At most two slow agents are at the same node.
	\item Only agents $g_1$ and $g_{2R}$ can become sentinel.
	\item Once a sentinel appears, none of the agents $g_2, \dots, g_{2R-1}$ is in a slow mode.
	\item Only fast agents can switch to mode $\idle$.
	\item A fast agent switches to mode $\idle$ only if it meets a sentinel.
	\item When $g_1$ meets $g_{2R}$, $g_1$ is in a fast mode and $g_{2R}$ is in mode $\senti$.
\end{enumerate}
Consider any node $u$ on the line, where a meeting happens in a round $t>z$.
Let $\Sigma$ denote the set of agents at node $u$ in round $t$.

It follows that 
\begin{enumerate}[label=\alph*)]
	\item $\Sigma$ contains at most two slow agents, in view of statement i).
	\item $\Sigma$ cannot contain a sentinel and a slow agent at the same time, in view of statements ii) and iii).
	\item $\Sigma$ contains at most one sentinel, in view of statement vi).
	\item If $\Sigma$ contains at least one agent in mode $\idle$, then $\Sigma$ has to contain a single sentinel, in view of statements iv), v) and vi).
	\item A slow agent cannot be in $\Sigma$ together with an agent in mode $\idle$, in view of statements b) and d).
\end{enumerate}

Therefore, cases 1-2 in {\tt Event D} are the only possible cases that involve slow agents; cases 1-4 are the only possible cases that involve fast agents; cases 4-5 are the only possible cases that involve agents in mode $\senti$ and $\idle$.
This completes the proof of the correctness of Algorithm {\tt Small Teams Unoriented}.

It remains to estimate the complexity of the algorithm.
Let $\maxDist(t)$ denote the distance between $g_1$ and $g_{2R}$ in round $t$.
For example, $\maxDist(0)$ is $D$ and $\maxDist(1)$ is $D+2$.

The following lemma gives an upper bound on $\maxDist(t)$ for any $t\leq t_2$ (where $t_2$ is the first round when both $g_1$ and $g_{2R}$ are already in mode {\tt sentinel}).

\begin{lemma}
	\label{lem-maxdist}
	$\maxDist(t)\le D+2\lceil t/z \rceil$, where $0\le t\le t_2$.
\end{lemma}
\begin{proof}
	Let $u$ and $v$ denote the bases of $g_1$ and $g_{2R}$, and let $u(t)$ and $v(t)$ denote the nodes that are reached in round $t$ by $g_1$ and $g_{2R}$, respectively.
	By definition, the distance between $u$ and $v$ is $D$.

	In the first $t_2$ rounds, each of $g_1$ and $g_{2R}$ is first in a slow mode and then in mode {\tt sentinel}.
	The dancing source of a slow agent moves one step in the progress direction of the agent every $z$ rounds.
	A sentinel stays put at its current node.
	In round $t$, $u(t)$ is left of $u$, at a distance at most $\lceil t/z\rceil$.
	Symmetrically, $v(t)$ is right of $v$, at a distance at most $\lceil t/z\rceil$.
	Therefore, $u(t)$ and $v(t)$ are at a distance at most $D+2\lceil t/z\rceil$.
	This proves the lemma.
\end{proof}

\begin{lemma}
	\label{lem-dist-t-1}
	$\maxDist(t_1)\le 2D+1$.
\end{lemma}
\begin{proof}
	Notice that $D+2\lceil t/z \rceil$ is a non-decreasing function of $t$.
	In view of Lemma \ref{lem-t1}, we have $t_1\leq \lceil D/2 \rceil z$.
	Hence, $\maxDist(t_1)\le D+ 2\lceil t_1/z \rceil \leq D+2\lceil D/2 \rceil \leq 2D+1$.
\end{proof}


The next lemma gives a bound on the round $t_2$ which is the first round when both $g_1$ and $g_{2R}$ are already in mode {\tt sentinel}.
\begin{lemma}\label{t_2}
	\label{lem-t-2}
	$t_2 \le \lceil D/2 \rceil z + 1344D+782$.
\end{lemma}

\begin{proof}
	Recall that in round $t_1$, both agents $g_1$ and $g_{2R}$ are in slow modes.
	Let $u'$ and $v'$ denote the nodes which agents $g_1$ and $g_{2R}$ reach in round $t_1$.
	By definition, the distance between $u'$ and $v'$ is $\maxDist(t_1)$.
	Let $g_i$ denote any fast agent between $u$ and $v$ in round $t_1$.
	Hence, we have $1<i<2R$.
	
	First assume that $g_i$ progresses left in round $t_1$.
	Let $\tau_i\ge t_1$ denote the round in which $g_i$ meets $g_1$.
	In view of Lemma \ref{lem-catch-slow-fast}, round $\tau_i$ exists.
	Notice that, in round $\tau_i$, the distance between $g_i$ and $g_{2R}$ is the same as the distance between $g_1$ and $g_{2R}$.
	The distance between $g_1$ and $g_{2R}$ in round $\tau_i$ is $\maxDist(\tau_i)$.
	In view of Lemma \ref{lem-catch-slow-fast}, $g_i$ meets $g_{2R}$ by round $\tau_i+48\cdot \maxDist(\tau_i)+7$.
	Hence, agent $g_i$ has met each of $g_1$ and $g_{2R}$ at least once by round $\tau_i+48\cdot \maxDist(\tau_i)+7$.
	Next, we compute the upper bound on $\tau_i+48\cdot \maxDist(\tau_i)+7$.
	
	\begin{claim}
		$\tau_i+48\cdot \maxDist(\tau_i)+7 \le \lceil D/2 \rceil z + 1344D+782$.
	\end{claim}
	
	To prove the claim, let $0\le x \le \maxDist(t_1)$ denote the distance between $g_1$ and $g_i$ in round $t_1$.
	In view of Lemma \ref{lem-catch-slow-fast}, we have $\tau_i\le t_1+48x+7$.
	In view of Lemma \ref{lem-t1}, we have $t_1\leq\lceil D/2 \rceil z$.
	In view of Lemma \ref{lem-maxdist}, we have $\maxDist(\tau_i)\le D+2\lceil \tau_i/z \rceil$.
	In view of Lemma \ref{lem-dist-t-1}, we have $x\le \maxDist(t_1)\le 2D+1$.
	Therefore, the upper bound of 	$\tau_i+48\cdot \maxDist(\tau_i)+7$ can be established as follows (using the fact that $z\geq 8$):
	\begin{align*}
		&\tau_i+48\cdot \maxDist(\tau_i)+7 \\
		\le & t_1+48x+7+48\cdot\maxDist(\tau_i)+7 \\
		< & \lceil D/2 \rceil z + 48x+7+48\cdot\maxDist(\tau_i)+7 \\
		\le & \lceil D/2 \rceil z + 48x+14+48\cdot(D+2\lceil \tau_i/z \rceil) \\
		\le & \lceil D/2 \rceil z + 48x+48D + 14 + 96 \lceil (t_1+48x+7)/z \rceil \\
		< & \lceil D/2 \rceil z + 48x+48D + 14 +96 (\lceil D/2 \rceil+6x+1) \\
		\le & \lceil D/2 \rceil z+48x+48D + 14 +48D+48+576x+96 \\
		= &\lceil D/2 \rceil z+624x+96D+158\\
		\le & \lceil D/2 \rceil z+624(2D+1)+96D+158\\
		=& \lceil D/2 \rceil z + 1344D+782
	\end{align*}
	This proves the claim. $\diamond$
	
	If $g_i$ progresses right in round $t_1$, then let $\tau'_i\ge t_1$ denote the round in which $g_i$ catches $g_{2R}$.
	Similarly, it follows that agent $g_i$ has met each of $g_1$ and $g_{2R}$ at least once by round $\tau'_i+48\cdot \maxDist(\tau'_i)+7$, where $\tau'_i+48\cdot \maxDist(\tau'_i)+7\le \lceil D/2 \rceil z + 1344D+782$.
	
	As $g_i$ is any fast agent between $u$ and $v$ in round $t_1$, it follows that by round $\lceil D/2 \rceil z + 1344D+782$, all the agents $g_{2}, \dots, g_{2R-1}$ have met each of $g_1$ and $g_{2R}$ at least once.
	Therefore, $t_2\le \lceil D/2 \rceil z + 1344D+782$.
\end{proof}

The next lemma estimates the distance between $g_1$ and $g_{2R}$ in round $t_2$.

\begin{lemma}\label{max-dist}
	$\maxDist(t_2)\le 338D + 197$.
\end{lemma}

\begin{proof}
	In view of Lemma \ref{lem-maxdist}, we have $\maxDist(t_2)\le D+2 \lceil t_2/ z \rceil$.
	In view of Lemma \ref{lem-t-2}, we have $t_2\le \lceil D/2 \rceil z + 1344D+782$.
	Therefore, in view of $z\geq 8$, we have
	\begin{align*}
		\maxDist(t_2)&\le D+2 \lceil t_2/ z \rceil \\
		&\le D+ 2\lceil (\lceil D/2 \rceil z + 1344D+782)/ z \rceil \\
		&\le D+2( \lceil D/2 \rceil + 168D +98) \text{, as $z\ge 8$} \\
		&\le D+D+1+336D+196\\
		&=338D + 197.
	\end{align*}
	This completes the proof of the lemma.
\end{proof}

The next lemma gives an upper bound on the round $t_3$ in which $2R-1$ agents gather at the same node.

\begin{lemma}
	\label{lem-t-3}
	$t_3\le t_2 + 16\cdot\maxDist(t_2)$
\end{lemma}
\begin{proof}
	Let $u''$ and $v''$ denote the nodes that agents $g_1$ and $g_{2R}$ reach in round $t_2$.
	Notice that from round $t_2$ to round $t_3$, both agents $g_1$ and $g_{2R}$ stay put at $u''$ and $v''$, respectively.
	So $\maxDist(t_2)=\maxDist(t)$ for any $t_2\le t \le t_3$.
	Recall that the dancing source of a fast agent moves one step in the progress direction of the agent every four rounds.
	Hence, within $4\cdot\maxDist(t_2)$ rounds, a fast agent can move from $u''$ to $v''$ or vice versa.
	
	Let $g_i$ denote any fast agent between $u''$ and $v''$ in round $t_2$.
	Hence, we have $1<i<2R$.
	Observe that, within the next $8\cdot\maxDist(t_2)$ rounds, $g_i$ has visited $g_1$ at least once, and by round $t_2 + 16\cdot\maxDist(t_2)$, $g_i$ has visited $g_1$ at least twice.
	By round $t_2$, both agents $g_1$ and $g_2$ have become sentinels.
	Each time $g_i$ visits a sentinel, $g_i.countSenti$ is incremented by one.
	It follows that there is a round $\tau_i$, where $t_2\le \tau_i\le t_2 + 16\cdot\maxDist(t_2)$, in which $g_i$ meets $g_1$ and $g_i.countSenti$ is increased to be at least $2$.
	In view of Lemma \ref{lem-comp-labels}, $g_i.compLabel$ stores the label of $g_{2R}$ at the meeting in round $\tau_i$.
	So, since $g_i.compLabel$ is greater than the label of $g_1$, $g_i$ switched to mode $\idle$ in round $\tau_i$.
	As $g_i$ is any fast agent between $u''$ and $v''$, it follows that by round $t_2 + 16\cdot\maxDist(t_2)$, all the agents $g_2, \dots, g_{2R-1}$ have switched to mode $\idle$ and stay put with $g_1$.
	So, $t_3\le t_2+16\cdot\maxDist(t_2)$.
	This completes the proof of the lemma.
\end{proof}

The next lemma gives an upper bound on the gathering round.

\begin{lemma}
	\label{lem-t-4}
	$t_4\le t_3 +4\cdot \maxDist(t_2) $
\end{lemma}
\begin{proof}
	In round $t_3$, all the agents $g_1, \dots, g_{2R-1}$ switch to a fast mode and progress towards $g_{2R}$.
	In view of Lemma \ref{lem-t-3-t-2}, $g_{2R}$ has stayed  put since $t_3$.
	As the distance between $g_1$ and $g_{2R}$ is  $\maxDist(t_2)$ in round $t_3$, agents $g_1, \dots, g_{2R-1}$ meet $g_{2R}$ within the next $4\cdot \maxDist(t_2)$ rounds.
	This completes the proof of the lemma.
\end{proof}

The final lemma establishes the order of magnitude of the gathering round.

\begin{lemma}\label{t_4}
	$t_4\in O(D\log L)$.
\end{lemma}
\begin{proof}
	In view of Lemmas \ref{lem-t-3} and \ref{lem-t-4}, we have $t_4\le t_2+ 20\cdot \maxDist(t_2)$.
	Therefore, in view of Lemmas \ref{t_2} and \ref{max-dist}, we have:
	\begin{align*}
		t_4&\le t_2+ 20\cdot \maxDist(t_2) \\
		&\le \lceil D/2 \rceil z + 1344D+782 + 20\cdot \maxDist(t_2) \\
		&\le \lceil D/2 \rceil z + 1344D+782 +20(338D+197) \\
		&\le \lceil D/2 \rceil z +8104D+4722 
	\end{align*}
	Recall that $z=2\cdot\len+4$ and $\len\in O(\log L)$.
	Hence, $t_4\in O(D\log L)$.
\end{proof}

Theorem \ref{unoriented-2} follows directly from Lemma \ref{t_4}.

\subsection{Teams of size larger than 2}
We now consider the gathering problem in any unoriented line for teams of equal size  $x$ larger than $2$.
Using more than two agents in each team, we design an algorithm faster than that for teams of size 2: our gathering algorithm for larger teams works in $O(D)$ rounds.
Obviously, this running time is optimal.
We present the algorithm and prove its correctness and complexity in the next two sections. 

\subsubsection{The algorithm}
\label{sect-alg-3-unoriented}
In this section, we present Algorithm {\tt Large Teams Unoriented}.
In our algorithm, we will use the following procedure {\tt Proceed} $(p, slow)$.
The procedure has two parameters:
The first parameter $p$ is a port number drawn from $\{-1, 1\}$, and the second parameter $slow$ is a Boolean value.
Like the procedure {\tt Dance}, {\tt Proceed} $(p, slow)$ is an infinite procedure divided into phases. 
Suppose that an agent $a$ starts performing the procedure {\tt Proceed} $(p, slow)$ at node $u$.

If $slow$ is false, then each phase has one round.
In the first phase, the agent moves to the neighbor $u'$ of $u$ corresponding to port $p$ at $u$.
For any phase $i>1$, let $w$ denote the node that the agent reached in the $(i-1)$-th phase.
In the $i$-th phase, the agent moves to the neighbor $w'$ of $w$ such that $dist(w', u)$ is $i$.
Intuitively, the agent moves away from $u$ with speed one in the direction of port $p$ at node $u$.

If $slow$ is true, then each phase consists of two consecutive rounds.
In the first phase, the agent moves to the neighbor $u'$ of $u$ corresponding to port $p$ at $u$, and stays at $u'$ for one round.
For any phase $i>1$, let $w$ denote the node that the agent reached in the $(i-1)$-th phase.
In the $i$-th phase, the agent moves to the neighbor $w'$ of $w$ such that $dist(w', u)$ is $i$ and stays at $w'$ for one round.
Intuitively, the agent moves away from $u$ with speed one-half in the direction of port $p$ at node $u$.
The pseudo-code of this procedure is as shown below.

	\noindent
	{\bf Procedure} {\tt Proceed} $(p, slow)$
	
	\noindent
	$FirstPhase:= $ {\bf true};\\
	{\bf Repeat forever}\\
	\hspace*{1cm}{\bf if} $FirstPhase$ {\bf then}\\
	\hspace*{2cm} move to the neighbor of the current node corresponding to port $p$; \\
	\hspace*{2cm} $FirstPhase:=$ {\bf false};\\
	\hspace*{2cm} {\bf if} $slow$ {\bf then}\\
	\hspace*{3cm} stay put for one round; \\ 
	\hspace*{1cm}{\bf else}\\
	\hspace*{2cm} $q:=$ the port by which the agent entered the current node most recently;\\
	\hspace*{2cm} move to the neighbor of the current node corresponding to port $-q$; \\ 
	\hspace*{2cm} {\bf if} $slow$ {\bf then}\\
	\hspace*{3cm} stay put for one round; \\ 

%

\noindent
{\bf Modes.} Similarly to the two previous algorithms, during the execution of Algorithm {\tt Large Teams Unoriented}, each agent is in one of five modes, encoded in the states of the automaton.
Each mode corresponds to a different action, as shown in Table \ref{tab-modes-unoriented-large}.
Therefore, there are five different actions an agent can choose to take.
Four of these actions are to call procedure {\tt Proceed} (with different parameters), and one of them is to stay put while waiting for future events.

\begin{table}[!h]
	\centering
	\caption{\label{tab-modes-unoriented-large}The five modes in Algorithm {\tt Large Teams Unoriented} and the corresponding actions.}
	\begin{tabular}{ |c| c| } 
		\hline
		Mode & Action		\\
		\hline
		\hline
		$\negSlow$ & {\tt Proceed}		$(-1, true)$\\
		\hline
		$\plusSlow$ & {\tt Proceed}		$(1, true)$\\
		\hline
		$\negFast$ & {\tt Proceed}		$(-1, false)$\\
		\hline
		$\plusFast$ & {\tt Proceed}		$(1, false)$\\
		\hline
		$\senti$ & 	Stay put\\
		\hline
	\end{tabular}
\end{table}

For simplicity, we call agents in modes $\negFast$ and $\plusFast$ {\em fast} agents, agents in modes $\negSlow$ and $\plusSlow$ {\em slow} agents, and agents in mode $\senti$ {\em sentinel} agents.

\noindent
\textbf{The high-level idea of the algorithm.}
After waking up, in each team, the agent with the smallest label assigns itself mode $\negSlow$, the agent with the second smallest label assigns itself mode $\plusSlow$, and all the other $(x-2)$ agents assign themselves mode $\senti$.
Notice that the agents with the two smallest labels in each team proceed in different directions, and the agent with the largest label in each team always assigns itself mode $\senti$.

Slow agents move along the line with speed one-half.
We call any base other than its own a {\em foreign base} for an agent.
Each time a slow agent meets $(x-2)$ agents in mode $\senti$, it reaches a foreign base.
A slow agent can tell how many foreign bases it has reached, by counting the total number of sentinel agents it has seen so far.
During the execution of the algorithm, only two of the slow agents can reach $(R-1)$ foreign bases.
These two slow agents are from the two most distant teams.
Let $u$ denote the $(R-1)$-th foreign base that a slow agent $a$ reaches.
By the round when agent $a$ reaches $u$, it has met $(x-2)(R-1)$ agents in mode $\senti$.
Agent $a$ switches to mode $\senti$. 
The largest label $\ell_1$ among agents from the team of $a$ is compared to the largest label $\ell_2$ among agents based at $u$.
(These labels are coded in the states of agents currently collocated at $u$).
If $\ell_1<\ell_2$, then the agent with label $\ell_2$ becomes a fast agent. 
This agent, called $e$, has the largest label among all agents in the two most distant teams.
$e$ is the first agent to become fast.
When this happens, all the $x\cdot R$ agents are distributed on the line as follows: there are $(x-2)$ sentinel agents at node $u$, along with the fast agent $e$; there are only $(R-1)$ slow agents, on one side of $u$, progressing away from $u$; there are $(x-2)(R-1)+1$ sentinel agents and $(R-1)$ slow agents on the other side of $u$, and these slow agents are progressing away from $u$ as well.
Agent $e$ progresses in the direction (with respect to $u$) where only slow agents are located.

As a fast agent, agent $e$ moves with speed one.
The fast speed allows agent $e$ to meet all the $(R-1)$ slow agents that it follows.
Each time agent $e$ meets a slow agent, it increments the total number of slow agents it has seen since becoming fast, and the slow agent switches to the same mode as that of $e$ and moves together with agent $e$ from now on.
Let $v$ denote the node where agent $e$ meets the $(R-1)$-th slow agent.
In the round when agent $e$ reaches node $v$, there are $R$ agents at this node, including agent $e$ itself.
In this round, all these $R$ agents become fast agents and move in the direction that agent $e$ comes from.
At this point, all the $x\cdot R$ agents are distributed on the line as follows: there are $R$ fast agent at node $v$; there are no agents at all on one side of $v$; there are $(x-1)R-(R-1)$ agents in mode $\senti$ and $R-1$ slow agents on the other side of $v$, and these slow agents move away from $v$. Hence, all $R$ fast agents follow these $R-1$ slow agents.
Since then, each time a meeting happens at some node $w$, all the non-fast agents at node $w$ (i.e., slow agents and sentinel agents) switch to the mode shared by all fast agents at node $w$.
In this way, all the agents at node $w$ move together, after the meeting.
When a meeting happens at a node where $x\cdot R$ agents are together, all the agents transit to state {\tt STOP}.
These agents know when gathering occurs by counting the number of agents collocated at each meeting.

\noindent
\textbf{Detailed description of the algorithm.}
Algorithm {\tt Large Teams Unoriented} is based on an event driven mechanism, as are the two previous algorithms.
This time, we have two types of events:  wake-up (Event A) and meetings (Event B), which are exclusive.
In each round, each agent is in some mode.
When an event happens, agents possibly switch to new modes and the corresponding action is taken.
Below we describe in detail what happens at each event.

In the pseudo-code, each agent $a$ has a variable $a.maxLabel$ that is initialized in the wake-up round and stores the largest label of the agents from its team.
Each slow agent $b$ has a variable $b.countSenti$ that is initialized to 0 in the wake-up round and incremented in each round when $b$ meets agents in mode $\senti$. In every round, this variable stores the total number of agents in mode $\senti$ that agent $b$ has seen so far.

	\noindent
	{\bf Event A:} All $R$ teams are woken up. \\
	
	\noindent
	$a := $ the agent with the smallest label in the team; \\
	$b := $ the agent with the second smallest label in the team; \\
	$a$ assigns itself mode $\negSlow$; \\
	$a.maxLabel :=$ the largest label in the team; \\
	$a.countSenti := 0$;\\
	$b$ assigns itself mode $\plusSlow$; \\
	$b.maxLabel :=$ the largest label in the team; \\
	$b.countSenti := 0$;\\
	{\bf for} each agent $d$ in the team, other than $a$ and $b$\\
	\hspace*{1cm} $d$ assigns itself mode $\senti$; \\
	\hspace*{1cm} $d.maxLabel :=$ the largest label in the team; \\

	\noindent
	{\bf Event B:} A meeting happens at a node $u$.\\
	
	\noindent
	{\bf if} there exists at least one fast agent $f$ at node $u$ {\bf then}\\
	\hspace*{1cm} {\bf if} there are $R$ agents at node $u$ {\bf then}\\
	\hspace*{2cm} $dir:=$  the port  number by which $f$ reached $u$; \\
	\hspace*{2cm} {\bf if} $dir = 1$ {\bf then}\\
	\hspace*{3cm} all $R$ agents at $u$ switch to mode $\plusFast$; \\
	\hspace*{2cm} {\bf else}\\
	\hspace*{3cm} all $R$ agents at $u$ switch to mode $\negFast$; \\
	\hspace*{1cm} {\bf else}\\
	\hspace*{2cm} {\bf if} there are $x\cdot R$ agents at node $u$ {\bf then}\\
	\hspace*{3cm} each agent at node $u$ transits to state {\tt STOP}; \\
	\hspace*{2cm} {\bf else}\\
	\hspace*{3cm} each agent, other than $f$, at node $u$ switches to the mode of $f$;\\
	{\bf else}\\
	\hspace*{1cm} {\bf if} there exists at least one agent in mode $\senti$ at node $u$ {\bf then}\\
	\hspace*{2cm} {\bf if} there are at least two slow agents at node $u$ {\bf then}\\
	\hspace*{3cm} {\bf for} each slow agent $e$ at node $u$ \\
	\hspace*{4cm} $e.countSenti := e.countSenti +(x-2)$; \\
	\hspace*{2cm} {\bf else} (** there is exactly one slow agent at $u$ **)\\
	\hspace*{3cm} $g:=$ the unique slow agent at node $u$;\\
	\hspace*{3cm} $g.countSenti := g.countSenti+(x-2)$; \\
	\hspace*{3cm} {\bf if} $g.countSenti =(x-2)(R-1)$ {\bf then}\\
	\hspace*{4cm} $c:=$ the sentinel agent with the largest label at node $u$;\\
	\hspace*{4cm} $g$ switches to mode $\senti$; \\
	\hspace*{4cm} {\bf if} $c.maxLabel > g.maxLabel$ {\bf then}\\
	\hspace*{5cm} $dir:$= the port number by which $g$ reached $u$; \\
	\hspace*{5cm} {\bf if} $dir=1$ {\bf then} \\
	\hspace*{6cm} $c$ switches to mode $\negFast$; \\
	\hspace*{5cm} {\bf else}\\
	\hspace*{6cm} $c$ switches to mode $\plusFast$; \\

\noindent
{\bf Remark.}
Each meeting in the algorithm has to involve at least one moving agent.
More precisely, if {\tt Event B} happens at some node $u$ in some round, then there is at least one fast agent or one slow agent at node $u$ in this round.
Hence, {\tt Event B} cannot happen, if there are only $(x-2)$ sentinel agents at their base.
In the pseudo-code for {\tt Event B},  the following types of meetings are considered:

\begin{itemize}
	\item \texttt{type-1}: At least one fast agent at meeting;
	\begin{itemize}
		\item \texttt{type-1.1}: $R$ agents at the meeting;
		\item \texttt{type-1.2}: $x\cdot R$ agents at the meeting;
		\item \texttt{type-1.3}: $j$ agents at the meeting, where $j\ne R$ and $j\ne x\cdot R$;
	\end{itemize}
	\item \texttt{type-2}: no fast agents and at least one sentinel agent at the meeting;
	\begin{itemize}
		\item \texttt{type-2.1}: more than one slow agent at the meeting;
		\item \texttt{type-2.2}: exactly one slow agent at the meeting;
	\end{itemize}
	\item \texttt{type-3}: only slow agents at the meeting.
\end{itemize}

These are all the types of meetings that can happen during the execution of Algorithm {\tt Large Teams Unoriented}.
In particular, the meeting of type-1.2 is the gathering.
This completes the detailed description of the algorithm.

\subsubsection{Correctness and complexity}
\label{sect-proof-3-unoriented}

In this section, we present the proof of correctness of Algorithm {\tt Large Teams Unoriented}, and
analyze its running time.

\begin{theorem}
	\label{theorem-large-unoriented}
	The algorithm {\tt Large Teams Unoriented} gathers $R$ teams of $x>2$ agents at the same node of an unoriented line, in $O(D)$ rounds, where $D$ denotes the distance between the bases of the most distant teams.
\end{theorem}

\begin{proof}
	First, we prove the correctness of the algorithm.
	Consider all $R$ teams on the line from left to right and assign the integer $i$ to the $i$-th team from left, for each $1\le i\le R$.
	In view of the algorithm, from the wake-up round on, two agents in each team progress left and right, respectively, and the other $(x-2)$ agents in the team stay put at their base in mode $\senti$.
	Let $p_i$ and $q_i$ denote the agents progressing right and left, respectively, in the $i$-th team, for each $1\le i\le R$.
	In the wake-up round, there are $R-i$ bases right of agent $p_i$ and $R-i$ bases left of $q_i$. 
	So, only two slow agents $p_1$ and $q_R$ will find $(R-1)$ foreign bases, while all the other $(2R-2)$ slow agents will find fewer than $(R-1)$ foreign bases.

	Let $u_i$ denote the base of the $i$-th team for each $1\le i\le R$, and let $t_1$ denote the round when agent $p_1$ reaches node $u_R$.
	Symmetrically, in the same round, agent $q_R$ reaches node $u_1$.
	By round $t_1$, both agents $p_1$ and $q_R$ have met $(R-1)$ foreign bases.
	In view of the algorithm, both agents $p_1$ and $q_R$ switch to mode $\senti$ in round $t_1$.
	
	Without loss of generality, we assume in the sequel that the largest label of the agents in the $R$-th team is larger than the largest label of the agents in the first team.
	Let $c$ denote the agent that has the largest label in the $R$-th team.
	We show that in round $t_1$, only agent $c$ becomes fast.
	Indeed, in round $t_1$, agent $p_1$ reaches the base $u_R$ of $c$, and thus it meets agent $c$, while $c$ is in mode $\senti$.
	Let $\delta$ denote the port number by which agent $p_1$ reaches $u_R$.
	Since $p_1.maxLabel < c.maxLabel$, agent $c$ becomes fast and starts performing procedure {\tt Proceed$(-\delta, false)$} from now on, in view of the algorithm.
	Notice that agent $c$ reaches $u_R$ from the left neighbor of $u_R$, which means that port $-\delta$ is right of $u_R$, so agent $c$ moves right with speed one after round $t_1$. 
	In the same round, agent $q_R$ reaches the base $u_1$ of the first team.
	Since $q_R.maxLabel$ is larger than the largest label in the first team, no agents at node $u_1$ become fast after the meeting in round $t_1$, in view of the algorithm.
	Therefore, agent $c$ that has the largest label among all agents in the first and the $R$-th teams is the only agent becoming fast in round $t_1$.
	
	Moreover, only slow agents $p_2, \dots, p_R$ are right of $u_R$ in round $t_1$.
	Indeed, in the first $t_1$ rounds, all slow agents $p_1, \dots, p_R$ progress right with speed one-half.
	Each of these $R$ agents moves in the same round and stays put in the same round.
	Hence, agent $p_{i+1}$ is always right of agent $p_i$ in the first $t_1$ rounds, for any $1\le i< R$.
	Recall that $u_R$ is the base of the rightmost team.
	No sentinel agents are right of $u_R$.
	In round $t_1$, agent $p_1$ reaches node $u_R$ and all the slow agents $p_2, \dots, p_R$ are right of $u_R$.
	Each of these agents $p_2, \dots, p_R$ progresses right.
	Symmetrically, only slow agents $q_1, \dots, q_{R-1}$ are left of node $u_1$, i.e., of the base of the first team, in round $t_1$, 
	and each of these agents progresses left.
	
	In the wake-up round, there are $(x-2)(R-1)$ agents in mode sentinel left of node $u_R$.
	In round $t_1$, agent $q_R$ switches to mode $\senti$ and stays put at node $u_1$.
	Hence, there are $(x-2)(R-1)+1$ sentinel agents left of $u_R$.
	
	After round $t_1$, the fast agent $c$ follows the slow agents $p_2, \dots, p_R$.
	In view of the algorithm, each time agent $c$ meets a slow agent, the slow agent becomes fast and moves together with $c$.
	Let $t_2$ denote the round when agent $c$ meets agent $p_R$ and let $v$ denote the node where this meeting occurs.
	In round $t_2$, all agents, $p_2, \dots, p_R$ and $c$, are at node $v$.
	In view of the algorithm, after the meeting in round $t_2$, these $R$ agents move left with speed one.
	
	In round $t_2$, there are $(R-1)$ slow agents left of $v$, which are $q_1, \dots, q_{R-1}$.
	In the wakeup round, there are $(x-2)R$ sentinel agents between nodes $u_1$ and $u_R$.
	In round $t_1$, one of them becomes fast, while $q_R$ and $p_1$ switch to mode $\senti$.
	Therefore, there are $(x-2)R+1$ agents in mode $\senti$ between nodes $u_1$ and $u_R$ in round $t_1$, as well as in round $t_2$.
	In round $t_2$, all $R$ fast agents $p_2, \dots, p_R$ and $c$ are at node $v$ and progress left.
	The other $(x-1)R$ agents are left of $v$, and each of them is either a slow agent or a sentinel agent.
	In view of the algorithm, each time fast agents meet slow agents or sentinel agents, all the non-fast agents participating in the meeting become fast and move together with the other fast agents.
	So, all the $xR$ agents will meet and gather at the same node.
	This proves the correctness of the algorithm. 
	
	Next, we analyze the complexity of the algorithm.
	Since $dist(u_1, u_R)=D$ and agent $p_1$ moves from $u_1$ to $u_R$ with speed one-half, it follows that $t_1=2D$.
	Recall that agents $p_1, \dots, p_R$ move right with speed one-half, and each of these agents moves or stays in the same round.
	The distance between $p_1$ and $p_R$ does not change in the first $t_1$ rounds.
	Hence, their distance is $D$, during this period.
	As both agents $c$ and $p_1$ are at node $u_R$ in round $t_1$, the distance between $c$ and $p_R$ is also $D$ in this round.
	As $c$ progresses right with speed one and $p_R$ progresses right with speed one-half, $c$ meets $p_R$ in additional $2D$ rounds; hence, $t_2=4D$.
	In round $t_2$, agent $c$ reaches node $v$.
	Hence, $dist(v, u_R)=t_2-t_1=2D$.
	
	Recall that agents $q_1, \dots, q_{R-1}$ move left with speed one-half, and each of these agents moves or stays in the same round.
	Therefore, $q_1$ is always left of the agents $q_2, \dots, q_{R-1}$ in any round.
	It follows that $q_1$ is the leftmost agent and $c$ is one of the rightmost agents on the line in round $t_2$.
	So, gathering is achieved in the round when $c$ meets $q_1$.
	Recall that agent $q_1$ starts from node $u_1$ and moves left with speed one-half.
	In round $t_2$, agent $q_1$ reaches the node $w$ left of $u_1$ at a distance $t_2/2=2D$, as $t_2=4D$.
	As $dist(w, u_1)=2D$, $dist(u_1, u_R)=D$, and $dist(u_R, v)=2D$, the distance between agents $q_1$ and $c$ is $5D$ in round $t_2$. Thus agent $c$ meets agent $q_1$ in additional $10D$ rounds.
	Overall, gathering is achieved in round $t_2+10D=14D$.
	This concludes the proof.
\end{proof}

\section{Teams of different sizes}

While in the previous sections we focused on the case when all teams of agents are of equal size $x$, in this section we consider the case when there are at least two teams of different sizes. The main result of this section says that,  in this case, gathering can be always achieved in time  $O(D)$, where $D$ is the distance between most distant bases, even in the unoriented line.


Let $x_{min}$ and $x_{max}$ denote the sizes of the smallest and largest teams, respectively.
Thus $x_{min}<x_{max}$. 
We assume that both $x_{min}$ and $x_{max}$ are known to all agents.

\subsection{When $x_{max}=2$ and $x_{min}=1$}\label{sec-one-and-two}

In this section, we consider the special situation when $x_{max}=2$ and $x_{min}=1$.
Thus all teams are either of size 1 or 2, and there is at least one team of each of those sizes.
We present the gathering algorithm \texttt{Team Sizes One and Two} which will be shown to work in time $O(D)$.
Our algorithm applies the procedure \texttt{Proceed} introduced in Section \ref{sect-alg-3-unoriented}.

\noindent
\textbf{Modes.} During the execution of \texttt{Team Sizes One and Two}, each agent is in one of ten modes,
encoded in the states of the automaton. Each mode corresponds to a different action, as shown in Table \ref{tab-modes-unoriented-arbitary}.

\begin{table}[!h]
	\centering
	\caption{\label{tab-modes-unoriented-arbitary}The ten modes in Algorithm {\tt Team Sizes One and Two} and the corresponding actions.}
	\begin{tabular}{ |c| c| } 
		\hline
		Mode & Action		\\
		\hline
		\hline
		$\negSlow$ & {\tt Proceed}		$(-1, true)$\\
		\hline
		$\plusSlow$ & {\tt Proceed}		$(1, true)$\\
		\hline
		$\negFast$ & {\tt Proceed}		$(-1, false)$\\
		\hline
		$\plusFast$ & {\tt Proceed}		$(1, false)$\\
		\hline
		$\negColl$ & {\tt Proceed}		$(-1, false)$\\
		\hline
		$\plusColl$ & {\tt Proceed}		$(1, false)$\\
		\hline
		$(-1)$-\texttt{search} & {\tt Proceed}		$(-1, false)$\\
		\hline
		$(+1)$-\texttt{search} & {\tt Proceed}		$(1, false)$\\
		\hline
		$\hiber$ & 	Stay put\\
		\hline
		$\idle$ & 	Stay put\\
		\hline
	\end{tabular}
\end{table}

\noindent
\textbf{The high level idea of the algorithm.}
We call an agent in a team of size one a \emph{singleton agent}, and its team a \emph{singleton team}.
We call both agents in a non-singleton team \emph{non-singleton agents}.
If $b$ is an agent in a non-singleton team, then the other agent in this team is called its {\em companion}.
After wake-up, the agent with the smaller label in every non-singleton team assigns itself the mode $\negSlow$, while the other agent in the team assigns itself the mode $\plusSlow$.
Intuitively, both agents move away from their common base with speed one half in different directions.
Every singleton agent assigns itself the mode $\idle$ and stays idle after wake-up, waiting for meetings with non-singleton agents in \texttt{slow} modes.
Every non-singleton agent maintains a variable, $bag$, that is initially set to empty and later used to store the distinct labels of singleton agents. 
Every singleton agent also maintains a variable, $bag$, that is initially set to empty and later used to store the distinct labels of non-singleton agents.

If a singleton agent in the mode $\idle$, say $a$, finds at a meeting only one non-singleton agent in a \texttt{slow} mode, say $b$, then 
$a$ switches to a \texttt{fast} mode and swings between $b$ and the companion $b'$ of $b$, with full speed: Specifically, agent $a$ moves towards the agent $b'$ with speed one; after it \emph{captures} $b'$, $b'$ appends to its bag the label of $a$ if the label is not stored in the bag of $b'$, $a$ appends to its bag the labels of $b$ and $b'$ if any of these labels are not stored in the bag of $a$, and $a$ moves towards $b$, and so on.
We say that an agent $x$, performing \texttt{Proceed($\cdot, false$)}, captures another agent $y$, performing \texttt{Proceed($\cdot, true$)}, if both $x$ and $y$ progress in the same direction and $x$ meets $y$, while $x$ is moving and $y$ is staying idle at the meeting round.
We say that the agent $a$ that swings between $b$ and $b'$ is \emph{associated} with the team of $b$ and $b'$.
While swinging between $b$ and $b'$, if $a$ (in a \texttt{fast} mode) {captures} any non-singleton agent in a \texttt{slow} mode, say $c$, then $a$ appends to its bag the labels of $c$ and of the team companion $c'$ of $c$, if any of these labels is not stored in the bag of $a$, and $c$ appends to its bag the label of $a$ if this label is not stored in the bag of $c$.

If a singleton agent in the $\idle$ mode is met in the same round by two non-singleton agents in \texttt{slow} modes, then the singleton agent is associated with the non-singleton agent  whose label is smaller between these two non-singleton agents.

Let $\alpha$ denote the singleton agent whose label is the smallest among all singleton agents.
We will prove in Lemma \ref{lem-alpha-switch-search} that in $O(D)$ rounds, the singleton agent $\alpha$ meets a non-singleton agent, say $b$, and at the meeting, $\alpha$ stores in its bag the label of every non-singleton agent, and $b$ stores in its bag the label of every singleton agent.
Both agents $\alpha$ and $b$ can be aware of such a meeting by checking if or not $|\alpha.bag|/2+|b.bag|=R$.
Hence, in some round $t\in O(D)$, the agent $\alpha$ knows all labels of singleton agents by reading the set $b.bag$, the total number $X_s$ of singleton agents, and the total number $2(R-X_s)$ of non-singleton agents.

Let $\ell_R$ (resp. $r_L$) denote the agent that moves left (resp. right) in the non-singleton team of which the base is rightmost (resp. leftmost).
Note that neither $\ell_R$ nor $r_L$ can tell left or right, since all agents are on an unoriented line.
We will prove in Lemma \ref{lem-progressing-R-X-s} that in the earliest round, denoted by $t'$ ($\le t$), when any singleton agent $a$ captures a non-singleton agent $c$ and $a.bag$ stores $2(R-X_s)$ agents at the meeting, $r_L$ is already right of $\ell_R$ and there are exactly $R-X_s$ non-singleton agents in a \texttt{slow} mode on the other direction, say $\mathcal{L}$, of the line than the progressing direction shared by $a$ and $c$.

When a singleton agent, say $a$, in a \texttt{fast} mode captures a non-singleton agent, say $b$, in a \texttt{slow} mode and finds that $\frac{|a.bag|}{2}+|b.bag|=R$, if the label of $a$ is not the smallest in $b.bag$, then $a$ maintains its mode unchanged and continues swinging between the team companions of its associated team; otherwise, $a$ recognizes itself as $\alpha$, switches to a \texttt{search} mode and moves towards the nearest non-singleton agent, say $d$, in the direction $\mathcal{L}$.
After $\alpha$ captures $d$, $\alpha$ switches to the mode $\hiber$ and waits for all other $X_s-1$ singleton agents.

Observe that $d$ must be either $\ell_R$ or $r_L$.
Recall that in the round when $\alpha$ switches to the mode $\hiber$,  $r_L$ is right of $\ell_R$.
Hence, the singleton agent $\alpha$ stays idle at a node of the line between $\ell_R$ and $r_L$.

When the agent $\alpha$ switches to the mode $\hiber$, any other singleton agent, say $c$, still swings.
Observe that when the agent $\ell_R$ is left of the agent $r_L$, both of them are between two non-singleton agents in every team. 
While $c$ is swinging between agents of  its associated team, $c$ captures both $\ell_R$ and $r_L$, due to the observation above.
Therefore, $c$ meets the agent $\alpha$.
After $c$ meets $\alpha$, $c$ switches to the mode $\hiber$ and stays idle with $\alpha$.
Within $O(D)$ rounds, all $X_s$ singleton agents gather at the same node.
Then, all these $X_s$ agents switch to a \texttt{collect} mode and progress in the direction $\mathcal{L}$.
Whenever they meet a non-singleton agent in a \texttt{slow} mode, the non-singleton agent also switches to the same \texttt{collect} mode as others and moves with them.

After they capture the $(R-X_s)$-th  non-singleton agent (i.e., the furthest non-singleton agent) in the direction $\mathcal{L}$, they switch to another \texttt{collect} mode and progress in the direction $\overline{\mathcal{L}}$, opposite to $\mathcal{L}$.
At this point, exactly $R$ agents, including $X_s$ singleton agents and $R-X_s$ non-singleton agents are at the same node, and there are other $R-X_s$ non-singleton agents in a \texttt{slow} mode, progressing in the direction $\overline{\mathcal{L}}$.
Within another $O(D)$ rounds, all $2R-X_s$ agents gather and transits to \texttt{STOP}.

\noindent
\textbf{Detailed description of the algorithm.}
Algorithm \texttt{Team Sizes One and Two} is based on an
event driven mechanism. 
There are seven pairwise exclusive events, including the wakeup event.
Below, we define each of these events and
describe in detail the actions to take at each event.

{\bf Event A} is the wake-up event.
In every singleton team, the singleton agent assigns itself the \texttt{idle} mode and creates an empty set $bag$ to store, later, the labels of non-singleton agents it has captured, as well as the labels of the team companions of these non-singleton agents.
In every non-singleton team, the agent with smaller label assigns itself the mode $\negSlow$, while the other agent assigns itself the mode $\plusSlow$.
Every non-singleton agent also creates a set, denoted by $bag$, to store, later, the labels of the singleton agents by which the non-singleton agent is captured.
Also, every agent in a non-singleton team stores the label of its companion in a variable $teamMem$.\\

\noindent
{\bf Event A:} All $R$ teams of agents are woken up.\\

\noindent
\textbf{if} there exists a single agent at the base \textbf{then}\\
\hspace*{1cm} $agent:=$ the single agent at the base;\\
\hspace*{1cm} $agent$ assigns itself mode $\idle$;\\
\hspace*{1cm} $agent.bag := \emptyset$;\\
\textbf{else}\\
\hspace*{1cm} $agent$-$sml:=$ the agent with smaller label;\\
\hspace*{1cm} $agent$-$sml$ assigns itself mode $\negSlow$;\\
\hspace*{1cm} $agent$-$lgr:=$ the agent with larger label;\\
\hspace*{1cm} $agent$-$lgr$ assigns itself mode $\plusSlow$;\\
\hspace*{1cm} $agent$-$sml.bag:=\emptyset$;\\
\hspace*{1cm} $agent$-$lgr.bag:=\emptyset$;\\
\hspace*{1cm} $agent$-$sml.teamMem:=$ $agent$-$lgr.label$;\\
\hspace*{1cm} $agent$-$lgr.teamMem:=$ $agent$-$sml.label$;\\

\textbf{Event B} is a meeting event, triggered by a meeting between an agent, say $a$, in mode $\idle$ and an agent (slow), say $b$, in the mode either $\plusSlow$ or $\negSlow$.
The agent $b$ maintains its mode and its bag unchanged after the meeting.
The agent $a$ stores, in the variable $goal$, the label of the team companion $b'$ of the agent $b$, maintains its bag unchanged, and switches from the mode $\idle$ to the mode $(port)$-\texttt{fast}, where $port$ is the port via which $b$ arrives at the meeting node.
Intuitively, the agent $a$ moves with speed one towards $b'$ after the meeting.\\

\noindent
\textbf{Event B:} A meeting between an agent $a$ in the $\idle$ mode and an agent $b$ in a \texttt{slow} mode. \\

\noindent
$a.goal := b.teamMem$;\\
$port:=$ the port via which $b$ arrives and meets $a$;\\
$a$ switches to mode $(port)$-\texttt{fast};\\
$b$ maintains its mode;\\

\textbf{Event C} is a meeting event, triggered by a meeting between an agent, say $a$, in the mode $\idle$ and two slow agents.
Let $b$ and $c$ denote, respectively, the slow agents with smaller and larger labels.
Both agents $b$ and $c$ maintain their respective modes and their respective bags unchanged after the meeting.
The agent $a$ sets its variable $goal$ to the label of the team companion of $b$, and maintains its bag unchanged.
Let $port$ denote the port via which $b$ arrives and meets $a$.
The agent $a$ switches to the mode $(port)$-\texttt{fast} after the meeting and moves towards the team companion of $b$ with speed one.\\

\noindent
\textbf{Event C:} A meeting between an agent $a$ in the $\idle$ mode and two agents in \texttt{slow} modes.\\

\noindent
$b:=$ the slow agent with smaller label;\\
$c:=$ the slow agent with larger label;\\
$a.goal := b.teamMem$;\\
$port:=$ the port via which $b$ arrives and meets $a$;\\
$a$ switches to mode $(port)$-\texttt{fast};\\
$b$ maintains its mode;\\
$c$ maintains its mode;\\

\textbf{Event D} is a meeting event, triggered when an agent (fast), say $a$, in the mode $(+1)$-\texttt{fast} or $(-1)$-\texttt{fast} {captures} a slow agent, say $b$.
Recall that we say that $a$ captures $b$ if $a$ moves and meets $b$, while $b$ stays, and $a$ and $b$ arrive at the meeting node via the same port.
At the meeting, $b$ appends to its bag the label of the agent $a$, if that label does not appear in $b$'s bag.
After the meeting, $b$ maintains its mode.

The agent $a$ appends to its bag the labels in $b$'s team.
Recall that $X_s$ denotes the number of singleton teams.
We say that an agent meets a team if it meets at least one agent of that team.
If $a$ meets $R-X_s$ non-singleton teams and $b$ meets $X_s$ singleton teams, i.e., $|a.bag|/2+|b.bag|=R$, and the label of $a$ happens to be the smallest in the set $b.bag$, then $a$ creates the variable $total$ that stores the total number of agents in all $R$ teams, i.e., $|a.bag|+|b.bag|$, creates the variable $countMultiTeams$ that stores the total number of non-singleton teams, i.e., $|a.bag|/2$, and switches to the mode $(port)$-\texttt{search}, where $port$ denotes the port via which $a$ arrives and meets $b$.
Otherwise, if the label stored in $a.goal$ is different from the label of $b$, then $a$ maintains its mode after the meeting; otherwise, $a$ updates its variable $goal$ to the label of $b$'s team companion and switches to the fast mode to move towards $b$'s team companion.\\

\noindent
\textbf{Event D:} An agent $a$ in a \texttt{fast} mode captures an agent $b$ in a \texttt{slow} mode.\\

\noindent
$b.bag:= b.bag \cup \{a.label\}$;\\
$b$ maintains its mode; \\
$port:=$ port via which $a$ arrives and meets $b$;\\
$a.bag := a.bag \cup \{b.label, b.teamMem\}$;\\
$min:=$ the smallest label in $b.bag$;\\
\textbf{if} $\frac{|a.bag|}{2}+|b.bag|= R$ and $a.label=min$ \textbf{then}\\
\hspace*{0.5cm} $a.total:=|a.bag|+|b.bag|$;\\
\hspace*{0.5cm} $a.countMultiTeams:=\frac{|a.bag|}{2}$;\\
\hspace*{0.5cm} $a$ switches to mode $(port)$-\texttt{search};\\
\textbf{else}\\
\hspace*{0.5cm}\textbf{if} $a.goal \ne b.label$ \textbf{then} \\
\hspace*{1cm} $a$ maintains its mode; \\
\hspace*{0.5cm}\textbf{else}\\
\hspace*{1cm} $a.goal:= b.teamMem$;\\
\hspace*{1cm} $a$ switches to mode $(port)$-\texttt{fast}; \\

\textbf{Event E} is a meeting event, triggered when an agent, say $a$, in a \texttt{search} mode captures an agent, say $b$, in a \texttt{slow} mode.
Note that there exists a unique agent that can switch to a \texttt{search} mode.
At the meeting, $a$ switches to the mode $\hiber$, while $b$ maintains its mode unchanged.\\

\noindent
\textbf{Event E:} An agent $a$ in a \texttt{search} mode captures an agent $b$ in a \texttt{slow} mode.\\

\noindent
$a$ switches to mode $\hiber$;\\
$b$ maintains its mode;\\

\textbf{Event F} is a meeting event, triggered when at least one agents in a \texttt{fast} mode meets a set of $\hiber$ agents.
One of the agents in the $\hiber$ set, denoted by $a$, stores the variables $countMultiTeams$ and $total$; both variables are created in \textbf{Event D}, and the agent $a$ has the smallest label among all singleton agents.
If the number of agents at the meeting is less than $X_s$, where $X_s=R-a.countMultiTeams$ denotes the total number of singleton teams, then every agent in mode \texttt{fast} switches to the mode $\hiber$ and stays idle with other agents after the meeting; otherwise, every agent switches to the mode $(-port)$-\texttt{collect}, where $port$ denotes the port via which $a$ arrived at the meeting node.\\

\noindent
\textbf{Event F:} At least one agent in \texttt{fast} mode meets a set of agents in the mode $\hiber$ at node $v$. \\

\noindent
$a:=$ the agent in the set that stores $total$; \\
$count:=$ the total number of agents at $v$;\\
\textbf{if} $count<R-a.countMultiTeams$ \textbf{then}\\
\hspace*{0.5cm} every fast agent at $v$ switches to $\hiber$;\\
\textbf{else}\\
\hspace*{0.5cm} $p:=$ the port via which $a$ arrives;\\
\hspace*{0.5cm} every agent at $v$ switches to $-p$-\texttt{collect};\\

\textbf{Event G} is a meeting event, triggered when a set of agents in the same mode, either $(+1)$-\texttt{collect} or $(-1)$-\texttt{collect}, captures a slow agent $b$.
One of the agents in the set, denoted by $a$, stores the variables $countMultiTeams$ and $total$. 
If the number of agents at the meeting is equal to the variable, $a.total$, which stores the total number of agents in all $R$ teams, then gathering is achieved and every agent transits to \texttt{STOP}.
Otherwise, $a$ decreases $countMultiTeams$ by one.
Recall that $countMultiTeams$ is initially set to the total number of non-singleton teams.
If $countMultiTeams$ is now equal to $0$, then $b$ is the last non-singleton agent in this progressing direction; so every agent, including $b$, switches to mode $(port)$-\texttt{collect}, where $port$ denotes the port via which the set of agents in the same \texttt{collect} mode arrive and meet $b$; otherwise, every agent, other than $b$, maintains its mode.
After the meeting, $b$ becomes an agent in a mode \texttt{collect} and moves together with other agents. \\

\noindent
\textbf{Event G:} A set of agents in the same $\texttt{collect}$ mode captures an agent $b$ in a \texttt{slow} mode at a node $v$.\\

\noindent
$a:=$ the agent that stores $total$; \\
$count:=$ the total num of agents at $v$;\\
\textbf{if} $a.total=count$ \textbf{then}\\
\hspace*{0.5cm} every agent at $v$ transits to \texttt{STOP};\\
\textbf{else}\\
\hspace*{0.5cm} $tmp:= a.countMultiTeams$; \\
\hspace*{0.5cm} $a.countMultiTeams:=tmp-1$; \\
\hspace*{0.5cm} \textbf{if} $a.countMultiTeams=0$ \textbf{then}\\
\hspace*{1cm} $p:=$ port via which $a$ arrives and meets $b$;\\
\hspace*{1cm} every agent at $v$ switches to $p$-\texttt{collect};\\ 
\hspace*{0.5cm} \textbf{else} \\
\hspace*{1cm} every agent ($\ne b$) at $v$ maintains its mode; \\
\hspace*{1cm} $b$ switches to mode of $a$ and moves with $a$;\\

%

This completes the detailed description of
Algorithm \texttt{Team Sizes One and Two}.
During the execution of the algorithm, a singleton agent in the mode $\idle$ can meet at most two slow agents in one round; this is because all $R$ teams are woken up simultaneously at different bases, and at most two slow agents can be at the same node.
Hence, no meetings between agents in the mode \texttt{idle} and slow agents, other than \texttt{Event B} and \texttt{Event C}, can happen.
Since there cannot exist more than one slow agent at the same node, progressing in the same direction, a \texttt{fast} agent can capture only one slow agent in one round, as shown by \texttt{Event D}, a \texttt{search} agent can capture only one slow agent in one round, as shown by \texttt{Event E}, and a set of agents in the same \texttt{collect} mode can capture only one slow agent in one round, as shown by \texttt{Event G}.
Finally, \texttt{Event F} is triggered by meetings between \texttt{fast} agents and $\hiber$ agents.
There are also meetings other than those listed above.
For instance, two slow agents can meet at the same node, while progressing in different directions.
When any of such meetings happens, each agent at the meeting maintains its mode and its variables unchanged.

\noindent
\textbf{Correctness and complexity.}
Recall that $X_s$ denotes the number of singleton teams and  that $\alpha$ denotes the singleton agent whose label is the smallest among all singleton agents.

The high-level idea of the proof is as follows:
We show in Lemmas \ref{lem-t-asso} and \ref{lem-t-x}, respectively, that gathering cannot be achieved until every singleton agent is associated with some non-singleton team or the agent $r_L$ is right of the agent $\ell_R$. 
In the remaining part of the proof, the following observation is crucial:

\begin{observation}\label{obser}
	While every non-singleton agent moves in a \texttt{slow} mode, if $r_L$ is right of $\ell_R$, then
	\begin{itemize}
		\item[i)] there are $R-X_s$ non-singleton agents that are progressing right (resp. left) and right (resp. left) of the singleton agent $b$, for any non-singleton agent $b$ progressing left (resp. right), and
		\item [ii)] there is at least one non-singleton agent from every non-singleton team, including $\ell_R$ and $r_L$, between (inclusive) a non-singleton agent and its companion of every non-singleton team.
	\end{itemize}
\end{observation}

Applying ii) of Observation \ref{obser}, we show in Lemma \ref{lem-alpha-switch-search} that while $\alpha$ is swinging between agents of the team associated with $\alpha$, it will capture a non-singleton agent such that, at the meeting, the union of the label sets stored in the bags of both agents is the set of labels of all $2R-X_s$ agents.
As a result, the agent $\alpha$ switches to a \texttt{search} mode at the meeting.
Applying i) of Observation \ref{obser}, we prove in Lemma \ref{lem-alpha-switch-hiber} that, while moving in a \texttt{search} mode with speed one, the agent $\alpha$ captures either $\ell_R$ or $r_L$ and switches to the mode $\hiber$ at the meeting.
Since then, the agent $\alpha$ stays in the mode $\hiber$ between $\ell_R$ and $r_L$, waiting for other $X_s-1$ singleton agents.
Applying ii) of Observation \ref{obser} again, we show in Lemma \ref{lem-hiber-all} that, while swinging between agents of its associated team, every other singleton agent meets $\alpha$ and switches to the mode $\hiber$.
Once all $X_s$ singleton agents gather at a node between $\ell_R$ and $r_L$, they switch to a \texttt{collect} mode.
When this happens, there are $R-X_s$ non-singleton agents, respectively, at each side of these $X_s$ singleton agents.
The set of singleton agents first capture $R-X_s$ non-singleton agents on one side of the line and then capture the remaining $R-X_s$ non-singleton agents on the other side of the line.
In the end, gathering is achieved after capturing either the leftmost or the rightmost non-singleton agent.

Let $t_{asso}$ denote the earliest round when every agent finds its associated team.
It might happen that different singleton agents are associated with the same team.

\begin{lemma}\label{lem-t-asso}
	The round $t_{asso}$ exists; $t_{asso}$ is an odd number; we have $t_{asso}\le 2D-1$; and by the round $t_{asso}$, the number of labels stored in the bag of every non-singleton is less than $X_s$. 
\end{lemma}
\begin{proof}
	In view of \texttt{Event D}, for any singleton agent $a$ and any non-singleton agent $b$, the label of $a$ is stored in the bag of $b$, if $a$ in a \texttt{fast} mode captures $b$ in a \texttt{slow} mode; recall that in the round when $a$ captures $b$, $a$ moves and $b$ stays.
	Therefore, in any round when some singleton agent $c$ switches from the mode $\idle$ to a \texttt{fast} mode, every non-singleton agent has not yet stored in its bag the label of $c$, and thus the agent $\alpha$ has not switched to a \texttt{search} mode. 
	Before $\alpha$ switches to a \texttt{search} mode, every singleton agent either finds its associated team or stays put in the mode $\idle$ at its base, and both non-singleton agents from each non-singleton team move away from their base with speed one half in a \texttt{slow} mode and progress in different directions.
	The distance between the bases of a singleton team and any non-singleton team is bounded by $D$.
	Hence, every singleton agent must meet some non-singleton agent, while the former stays put in the mode $\idle$.
	This proves the existence of $t_{asso}$.
	
	In the round $t_{asso}$, at least one non-singleton agent moves.
	Every agent in a \texttt{slow} mode moves only in odd rounds.
	Hence, $t_{asso}$ is an odd number.
	
	The distance between the bases of a singleton team and any non-singleton team is at most $D$; therefore, we have $t_{asso}\le 2D-1$.
	
	Let $f$ denote the singleton agent that switches from the mode $\idle$ to a \texttt{fast} mode in round $t_{asso}$.
	By the round $t_{asso}$, every non-singleton agent has not stored in its bag the label of $f$; hence, the bag size of every non-singleton agent is less than $X_s$ in round $t_{asso}$.
\end{proof}

Let $t_x$ denote the earliest round when the agent $r_{L}$ is right of $\ell_R$. 

\begin{lemma}\label{lem-t-x}
	The round $t_x$ exists; $t_x$ is an odd number; we have $t_x\le D+1$; and by the round $t_x$, the number of labels stored in the bag of every singleton agent is less than $2(R-X_s)$.
\end{lemma}
\begin{proof}
	If only one of $R$ teams is non-singleton, then both $r_L$ and $\ell_R$ belong to the same team, and $r_L$ is right of $\ell_R$ in the first round.
	Hence, $t_x$ exists, we have $t_x=1\le 2D$, and the bag size of every singleton agent is zero in the first round, which is fewer than $2(R-X_s)=2$.
	
	Otherwise, there are at least two non-singleton teams.
	Let $d$ denote the distance between the bases of the teams of  $r_L$ and $\ell_R$.
	If $d$ is odd, then $t_x$ (if it exists) equals to $2\frac{d-1}{2}+1=d$; otherwise, $t_x$ (if it exists) is $2\frac{d}{2}+1=d+1$.
	Hence, $t_x$  (if it exists)  is always an odd number.
	Let $a$ denote any singleton agent.
	If $a$ is not yet associated with any team in the round $t_x$, then the bag of $a$ is empty;
	otherwise, while $a$ is swinging between the two agents in its associated team, either $a$ has not caught $r_L$ and its companion or $a$ has not caught $\ell_R$ and its companion, by round $t_x$; therefore, the bag of $a$ stores less than $2(R-X_s)$ labels.
	This implies that the agent $\alpha$ has not switched to a \texttt{search} mode by the round $t_x$, if $t_x$ exists.
	Before the agent $\alpha$ switches to a \texttt{search} mode, the agent $r_L$ (resp. $\ell_R$) progresses right (resp. left) in a \texttt{slow} mode.
	Since the distance between the bases of $r_L$ and $\ell_R$ is upper bounded by $D$, $r_L$ must be right of $\ell_R$ in some round before $\alpha$ switches to a \texttt{search} mode; this proves the existence of the round $t_x$.
	Since $d\le D$, we have $t_x\le D+1$.
\end{proof}


\begin{lemma}\label{lem-progressing-R-X-s}
	In the round when a singleton agent, say $a$, whose bag stores $2(R-X_s)$ labels captures a non-singleton agent, progressing left (resp. right), there are $R-X_s$ non-singleton agents that are right (resp. left) of $a$ and progressing right (resp. left). 
\end{lemma}
\begin{proof}
	In view of Lemma \ref{lem-t-x}, only after the round $t_x$, can a singleton agent whose bag stores $2(R-X_s)$ labels capture a non-singleton agent.
	From the round $t_x$ on, $r_L$ is right of $\ell_R$.
	This implies the lemma, in view of  Observation \ref{obser} i).
\end{proof}


Lemma \ref{lem-basic-capture} shows in how many rounds an agent $a$ that moves with speed one can capture an agent that moves with half speed and progresses in the same direction as $a$.

\begin{lemma}\label{lem-basic-capture}
	Let $a$ denote the agent that performs 
	\texttt{proceed}($\cdot, false$) and let $b$ denote the agent that starts performing \texttt{proceed}($\cdot, true$) after wake-up.
	If $a$ and $b$ progress left (resp. right) and $a$ is left (resp. right) of $b$ at a distance $\ell$ in the round $t$, for any even number $t\ge 2$, then $a$ captures $b$ in round $t+2\ell$.
\end{lemma}

\begin{proof}
	Without loss of generality, assume that both $a$ and $b$ progress left.
	Hence, in round $t$, $a$ is left of $b$ at a distance $\ell$.
	Between rounds $t$ and $t+2\ell$, $a$ moves for $2\ell$ rounds and $b$ moves for $\ell$ rounds.
	Before the round $t+2\ell$, $a$ is left of $b$, and in the round $t+2\ell$, both $a$ and $b$ are at the same node.
	Moreover, since $t$ is even, $t+2\ell$ is also even.
	Therefore, in round $t+2\ell$, $a$ moves and $b$ stays idle, and thus $a$ meets $b$ in round $t+2\ell$.
\end{proof}

Lemma \ref{lem-double-capture} shows how many rounds it takes a singleton agent $a$ to capture every agent in the team associated with $a$  at least once, as $a$ swings.

\begin{lemma}\label{lem-double-capture}
	Let $a$ denote any singleton agent swinging in round $t$ between two non-singleton agents $b$ and $b'$ in the team associated with $a$.
	\begin{itemize}
		\item[i)] If $t$ is even, then $a$ catches each of $b$ and $b'$ at least once, between  round $t$ and round $9t$;
		\item[ii)] If $t$ is odd, then $a$ catches each of $b$ and $b'$ at least once, between round $t$ and round $9t+3$.
	\end{itemize}
\end{lemma}

\begin{proof}
	Without loss of generality, assume that in round $t$, $a$ is progressing towards $b$.
	
	If $t$ is even, then the distance between $b$ and $b'$ in round $t$ is $t$; hence, the distance between $a$ and $b$ is at most $t$.
	In view of Lemma \ref{lem-basic-capture}, $a$ captures $b$ by the round $t+2t=3t$.
	After capturing $b$, $a$ moves towards $b'$.
	By the round $3t$, the distance between $b'$ and $b$ is at most $3t$.
	In view of Lemma \ref{lem-basic-capture}, $a$ captures $b'$ by the round $3t+2(3t)=9t$.
	This proves the statement i).
	
	If $t$ is odd, then the distance between $b$ and $b'$ in round $t$ is $t+1$; hence, the distance between $a$ and $b$ is at most $t+1$.
	In round $t+1$, the distance between $a$ and $b$ is at most $t$.
	In view of Lemma \ref{lem-basic-capture}, $a$ captures $b$ by the round $t+1+2(t)=3t+1$.
	After capturing $b$, $a$ moves towards $b'$.
	By the round $3t+1$, the distance between $b'$ and $b$ is at most $3t+1$.
	In view of Lemma \ref{lem-basic-capture}, $a$ captures $b'$ by the round $3t+1+2(3t+1)=9t+3$.
	This proves the statement ii).
\end{proof}



Lemma \ref{lem-assumption-forever} is used later to prove that in some round the agent $\alpha$ captures a non-singleton agent such that at the meeting the former agent stores in its bag $2(R-X_s)$ labels, and the later agent stores in its bag $X_s$ labels.

\begin{lemma}\label{lem-assumption-forever}
	Assume that after switching to a \texttt{fast} mode, every singleton agent swings for at least  $54D+9$ rounds  between the two agents in its associated team. By the round $54D+9$,
	\begin{itemize}
		\item  every singleton agent has caught at least one non-singleton agent in every non-singleton team, including the agents $\ell_R$ and $r_L$, and
		\item both $\ell_R$ and $r_L$ have been caught, respectively, by every singleton agent.
	\end{itemize}
\end{lemma}
\begin{proof}
	%
	Let $t_m=\max(t_{asso}, t_x)$.
	In view of Lemmas \ref{lem-t-asso} and \ref{lem-t-x}, both $t_{asso}$ and $t_x$ are odd, and thus $t_m$ is odd.
	Let $t^a\ge t_m$, for any singleton agent $a$, denote the earliest round, not earlier than $t_m$, when $a$ captures a non-singleton agent, say $b$, in the team associated with $a$.
	
	Recall that $b$ stays idle in even rounds, and thus $t^a$ must be an even number.
	In view of statement i) of Lemma \ref{lem-double-capture}, from the round $t^a$ ($\ge t_m\ge t_x$) to the round $9t^a$, it happens at least once that 1) the singleton agent $a$ moves towards $b'$, starting from $b$, and captures $b'$, 2) and then the singleton agent $a$ moves towards $b$, starting from $b'$, and captures $b$.
	Therefore, by the round $9t^a$, the singleton agent $a$  has caught at least one non-singleton agent in every non-singleton team, including $\ell_R$ and $r_L$, in view of Observation \ref{obser} ii).
	%
	%

	Next, we show an upper bound on $t_a$.
	
	\begin{claim}\label{claim-t-a-upper-bound}
		We have $t^a\le 6D+1$.
	\end{claim}
	
	In order to prove the claim, recall that $t_m$ is an odd number and that $a$ captures $b$ when $a$ moves and $b$ stays.
	Therefore, $a$ cannot capture $b$ in the round $t_m$.
	Instead, $a$ is moving towards $b$ in round $t_m$, in view of the definition of $t^a$.	
	Let $b'$ denote the companion of the agent $b$ in the team associated with $a$.
	
	In round $t_m$, the distance between $b$ and $b'$ is $2\lceil \frac{t_m}{2}\rceil=t_m+1$, since $t_m$ is odd.
	Hence, the distance between $b$ and $a$ is at most $t_m+1$.
	In round $t_m+1$, $a$ moves and $b$ stays, as $t_m+1$ is even, and thus the distance between $b$ and $a$ is at most $t_m$.
	In view of Lemma \ref{lem-basic-capture}, $a$ has caught $b$ by the round $t_m+1+2(t_m)=3t_m+1$, and thus we have $t^a\le 3t_m+1$.
	In view of Lemmas \ref{lem-t-asso} and \ref{lem-t-x}, $t_m\le 2D$, and thus we have $t^a\le 6D+1$. $\diamond$

	In view of Claim \ref{claim-t-a-upper-bound}, we have $9t^a\le 54D+9$, for any singleton agent $a$.
	This concludes the proof.
\end{proof}

In view of \texttt{Event D}, only the agent $\alpha$ can switch to a \texttt{search} mode.
Next, we prove that $\alpha$ must switch to a \texttt{search} mode and give an upper bound on the round when it happens.

\begin{lemma}\label{lem-alpha-switch-search}
	The agent $\alpha$ switches to a \texttt{search} mode before the round $486D+84$.
\end{lemma}
\begin{proof}
	If the agent $\alpha$ switches to a \texttt{search} mode before the round $54D+9$, then the statement of the lemma is true.
	Otherwise, in the round $54D+9$, every singleton agent stores in its bag $2(R-X_s)$ labels of all non-singleton agents, and both $\ell_R$ and $r_L$ store in their respective bags the set of labels of all $X_s$ singleton agents, in view of Lemma \ref{lem-assumption-forever}.
	Hence, from the round $54D+9$ on, once the agent $\alpha$ captures one of non-singleton agents whose bag stores $X_s$ distinct labels, e.g., $\ell_R$ or $r_L$, $\alpha$ is guaranteed to switch to a \texttt{search} mode.
	
	Since $t_{asso}<2D$, the agent $\alpha$ has found its associated team by the round $54D+9$.
	Let $b$ and $b'$ denote two non-singleton agents in the team associated with $\alpha$.
	Without loss of generality, assume that $\alpha$ is progressing towards $b$ in round $54D+9$.
	Observe that $54D+9$ is an odd number.
	In view of statement ii) of Lemma \ref{lem-double-capture}, from the round $54D+9$ to the round $9(54D+9)+3$, it happens at least once that the agent $\alpha$ moves towards $b'$, starting from $b$, and captures $b'$.
	Recall that from the round $54D+9$ on, both $\ell_R$ and $r_L$ are between $b$ and $b'$.
	Therefore, $\alpha$ captures one non-singleton agent $c$ before the round $9(54D+9)+3$ such that at the meeting $\alpha$ stores $2(R-X_s)$ labels and $c$ stores $X_s$ labels.
	As a result, $\alpha$ switches to a \texttt{search} mode at the meeting, before the round $9(54D+9)+3=486D+84$.
\end{proof}

In view of \texttt{Event E}, the agent $\alpha$ switches from a \texttt{search} mode to the mode $\hiber$, once it captures an agent in a \texttt{slow} mode.
Next, we show an upper bound on the round when the agent $\alpha$ switches to the mode $\hiber$.

\begin{lemma}\label{lem-alpha-switch-hiber}
	The agent $\alpha$ has switched to the mode $\hiber$ by the round $1458D+252$.
\end{lemma}
\begin{proof}
	Recall that in the round when the agent $\alpha$ switches from a \texttt{fast} mode to a \texttt{search} mode, $\alpha$ changes its progressing direction and the bag stored by $\alpha$ has exactly $2(R-X_s)$ labels. 
	In view of Lemma \ref{lem-progressing-R-X-s}, there are $R-X_s$ non-singleton agents ahead of $\alpha$, each progressing in the same direction as $\alpha$.
	Hence, the first non-singleton agent caught by $\alpha$ must be either $r_L$ or $\ell_R$.
	
	If $\alpha$ switches to the mode $\hiber$ before the round $486D+84$, then the lemma holds.
	Otherwise, let $b$ and $b'$ denote the two non-singleton agents in the team associated with $\alpha$.
	Without loss of generality, assume that the agent $\alpha$ moves towards $b$ in the round $486D+84$.

	Observe that $486D+84$ is even and that in the round $486D+84$, the distance between $b$ and $b'$ is exactly $486D+84$.
	Since $\alpha$ is between $b$ and $b'$, the distance between $\alpha$ and $b$ is at most $486D+84$.
	In view of Lemma \ref{lem-basic-capture}, the agent $\alpha$ captures $b$ by the round $486D+84+2(486D+84)=1458D+252$.
	Recall that from round $t_x$ ($\le D+1$) on, both $\ell_R$ and $r_L$ are between (inclusive) the non-singleton agents in every non-singleton team.
	Therefore, the agent $\alpha$ captures either $\ell_R$ or $r_L$ by the round $1458D+252$, and thus $\alpha$ switches to the mode $\hiber$ by the round $1458D+252$.
\end{proof}

In view of \texttt{Event F}, when a singleton agent in a \texttt{fast} mode meets the agent $\alpha$, while the latter is in the mode $\hiber$, the former switches to the mode $\hiber$ as well. 
Next, we prove that every singleton agent switches to the $\hiber$ mode and stays with $\alpha$.

\begin{lemma}\label{lem-hiber-all}
	Every singleton agent, other than $\alpha$, meets $\alpha$ and switches to the mode $\hiber$ by the round $13122D+2268$.
\end{lemma}
\begin{proof}
	In view of Lemmas \ref{lem-t-asso} and \ref{lem-t-x}, the agent $\alpha$ switches to the $\hiber$ mode, after the round $\max(t_x, t_{asso})$.
	From the round $t_{asso}$ on, every singleton agent is associated with a non-singleton team.
	From the round $t_x$ on, the agents $\ell_R$ and $r_L$ are between two non-singleton agents in every non-singleton team.
	As the agent $\alpha$ switches to the mode $\hiber$ at a node between $\ell_R$ and $r_L$ and any singleton agent $a$, other than $\alpha$, swings between two non-singleton agents in the team associated with $a$, $a$ must meet $\alpha$ and switches to the mode $\hiber$.
	
	Assume that there are still singleton agents swinging in the round $1458D+252$; otherwise, the lemma holds.
	Let $c$ denote any of these singleton agents and let $b$ and $b'$ denote the non-singleton agents in the team associated with $c$.
	Without loss of generality, assume that $c$ is progressing towards $b$ in the round $1458D+252$.
	Observe that $1458D+252$ is even.
	In view of statement i) of Lemma \ref{lem-double-capture}, from the round $1458D+252$ to the round $9(1458D+252)$, it happens at least once that the agent $\alpha$ moves towards $b'$, starting from $b$, and captures $b'$.
	Since the agent $\alpha$ stays idle between $b$ and $b'$, $c$ must meet $\alpha$ and switch to the $\hiber$ mode, by the round $9(1458D+252)=13122D+2268$.
	This concludes the proof.
\end{proof}

In view of \texttt{Event F}, after all $X_s$ singleton agents are at the same node in the mode $\hiber$, each of them switches to a \texttt{collect} mode and moves in the same direction.
Without loss of generality, assume that they are progressing right.
Recall that $X_s$ singleton agents switched to the mode $\hiber$ at the same node between the agents $\ell_R$ and $r_L$.
Hence, there are $R-X_s$ non-singleton agents, progressing right in a \texttt{slow} mode, ahead of every non-singleton agent.
In view of \texttt{Event G}, each time the set of singleton agents capture a non-singleton agent, the latter switches to the same \texttt{collect} mode as others at the meeting and the agent $\alpha$ decrements its variable $countMultiTeams$ by one.
Recall that $countMultiTeams$ initially stores $R-X_s$.
When $countMultiTeams$ becomes $0$, agents capture the furthest non-singleton agent on the right side and they move together left with speed one.
Finally, all the agents gather after capturing the leftmost non-singleton agent.
The agent $\alpha$ knows when gathering happens, since $\alpha$ stores the total number of agents in its variable $total$.
This proves that algorithm \texttt{Team Sizes One and Two} guarantees gathering.
Its complexity is summarized as Theorem \ref{theorem-1-2}.

\begin{theorem}\label{theorem-1-2}
	Algorithm {\tt Team Sizes One and Two} gathers $R$ teams of agents at the same node of an unoriented line in $O(D)$ rounds, where $D$ denotes the distance between the bases of the most distant teams.
\end{theorem}
\begin{proof}
	Let $t_{coll}$ denote the round when all $X_s$ singleton agents switch from the $\hiber$ mode to a \texttt{collect} mode.
	Let $p$ (resp. $q$) denote the rightmost (resp. leftmost) non-singleton agent in the round $t_{coll}$.
	
	Let $t_1$ denote the round when $\alpha$ captures $p$.
	Note that $t_1$ must be even.
	In round $t_1$, the distance between $p$ and $q$ is at most $t_1+D$, as the distance between $p$ and its team companion $p'$ is $t_1$ and the distance between $p'$ and $q$ is at most $D$.
	In view of Lemma \ref{lem-basic-capture},  $\alpha$ captures $q$ by the round $t_1+2(t_1+D)=3t_1+2D$.
	Therefore, gathering is achieved by the round $3t_1+2D$.
	
	The distance between the agent $\alpha$ and $p$ in the round $t_{coll}$ is at most $t_{coll}+D$, as the distance between $\alpha$ and $r_L$ is at most $t_{coll}$ and the distance between $r_L$ and $p$ is at most $D$.
	
	In view of Lemma \ref{lem-basic-capture}, if $t_{coll}$ is even, then $t_1\le t_{coll}+2(t_{coll}+D)=3t_{coll}+2D$;
	otherwise, $t_1\le t_{coll}+1+2(t_{coll}+D-1)=3t_{coll}+2D-1$.
	Hence, we have $t_1\le 3t_{coll}+2D$.
	In view of Lemma \ref{lem-hiber-all}, we have $t_{coll}\leq 13122D+2268$, and thus gathering is achieved by the round $3t_1+2D\le 9t_{coll}+8D \leq118106D+20412$.
	This concludes the proof.
\end{proof}

%
%
%



\subsection{When $x_{max}>x_{min}$ are arbitrary distinct positive integers}
In this section, we assume that $x_{max}>x_{min}$ are arbitrary distinct positive integers.
This is the general case of teams of non-equal sizes.

We show that in this general case gathering can also be achieved in time $O(D)$. The proof is by a reduction to Algorithm {\tt Team Sizes One and Two}  presented in Section \ref{sec-one-and-two}.
Let $x_{avg}=(x_{max}+x_{min})/2$.
Note that $x_{avg}$ does not have to be an integer.
For each team on the unoriented line, if the team size is at most $x_{avg}$, then we regard the team as a singleton team; namely, all agents in the team behave exactly the same throughout the algorithm and are virtually regarded as one agent; otherwise, we regard the team as a non-singleton team composed of two agents; specifically, the agent with the smallest label in the team is regarded virtually as a non-singleton agent that calls 
the procedure \texttt{Proceed}($-1$, true) after wake-up, and all the other agents in the team are regarded virtually as the other non-singleton agent; namely, they call the procedure \texttt{Proceed}($1$, true) after wake-up, and behave exactly the same throughout the algorithm.
As a result, every team either has one virtual agent or two virtual agents.
Thus, we can apply Algorithm  \texttt{Team Sizes One and Two} to achieve gathering in time $O(D)$, where $D$ denotes the distance between the bases of the most distant teams. This implies the following theorem that completes our solution of the gathering problem.

\begin{theorem}
	If there are at least two teams of different sizes, then all the agents can be gathered in an unoriented line in $O(D)$ rounds, where $D$ denotes the distance between the bases of the most distant teams.
\end{theorem}

\section{Discussion}

In this section, we discuss the necessity of two assumptions made in our model. 
The first assumption concerns the knowledge of the size $L$ of the space of labels and of the number $R$ of teams, and the second is about  the simultaneous start of all agents.

The necessity of knowing $L$ is straightforward: the Mealy automaton formalizing the agents must have sufficiently many states to code the labels of the agents, in order to permit interaction between them at the meetings.
Next, we argue that without the knowledge of $R$, the gathering problem cannot be solved. Suppose that  it can, and that some hypothetical algorithm induces all agents to gather at one node $u$ and stop in round $t$. Let $A$ be the set of these agents. Then the adversary can add another team whose base is at distance $2t+1$ from $u$. In the first $t$ rounds, none of the agents from the set $A$ could meet with any agent from the added team. Hence, in the augmented set of agents, agents from $A$ would behave identically as before and hence gather at $u$ and stop in round $t$. This would be incorrect, as the added agents could not be at $u$ in round $t$.

Finally, we argue that, if the start was not simultaneous, then no bound on gathering time could be established, even in the oriented line and even if time counted from the wakeup round of the last agent. Indeed, consider two teams of some size $x$. If all trajectories of all agents were bounded, then the adversary could choose the bases of the teams so far that gathering would never happen. 
So at least one agent, call it $a$, in one of these teams must travel arbitrarily far from its base. Suppose its trajectory is left-progressing.
Thus, before meeting an agent from the other team, agent $a$ can go only at a distance of at most $d$ right of its base.  Consider any positive integer $B$.
The adversary chooses the base of the team of $a$ at distance $d+1$, left of the other base, wakes up agent $a$ in round 0, and wakes up all other agents in round $t$ when agent $a$ is at distance $B$ left of its base. We start counting time in round $t$. However, in round $t$, agent $a$ is at distance $B$ from any other agent, so it takes at least $B/2$ rounds until any meeting involving agent $a$ can occur.

\section{Conclusion}

We gave a complete solution of the feasibility and complexity problem of gathering teams of agents modeled as deterministic automata, both on the oriented and on the unoriented infinite line, showing differences that arise from the orientation feature. To the best of our knowledge, this is the first time gathering of deterministic automata that cannot communicate remotely is considered in an infinite environment. For infinite lines, it is impossible to design an automaton large enough to explore the entire graph, and hence we had to invent new gathering techniques. A natural generalization of our results would be to study gathering teams of automata in arbitrary (connected) infinite graphs. 




\section*{Declarations}

A preliminary version of this article appeared at the 28th International Conference on Principles of Distributed Systems (OPODIS 2024). Compared to the conference version, this manuscript includes full proofs of all lemmas and theorems and a new section in which we present an optimal $O(D)$-time algorithm for gathering of teams of non-identical sizes.


\begin{thebibliography}{12}
	
	
	
	%
	%
	
	\bibitem{alpern02b}
	S. Alpern and S. Gal,
	The theory of search games and rendezvous.
	Int. Series in Operations research and Management Science,
	Kluwer Academic Publisher, 2002.
	
	\bibitem{BIOKM}
	D. Baba, T. Izumi, F. Ooshita, H. Kakugawa, T. Masuzawa, 
	Linear time and space gathering of anonymous mobile agents in asynchronous trees, 
	Theoretical Computer Science 478 (2013),  118-126.
	
	%
	%
	
	
	\bibitem{BCGIL}
	E. Bampas, J. Czyzowicz, L. Gasieniec, D. Ilcinkas, A. Labourel, Almost optimal asynchronous rendezvous in infinite multidimensional grids,
	Proc. 24th International Symposium on Distributed Computing (DISC 2010),  297-311.
	
	
	%
	
	%
	
	
	\bibitem{BBDDP}
	S. Bouchard, M. Bournat, Y. Dieudonn\'{e}, S. Dubois, F. Petit,
	Asynchronous approach in the plane: a deterministic polynomial algorithm, 
	Distributed Computing 32 (2019), 317-337.
	
	\bibitem{BDD}
	S. Bouchard, Y. Dieudonn\'{e}, B. Ducourthial,
	Byzantine gathering in networks,
	Distributed Computing 29 (2016), 435-457.
	
	\bibitem{BDL}
	S. Bouchard, Y. Dieudonn\'{e}, A. Lamani,
	Byzantine gathering in polynomial time,
	Distributed Computing 35 (2022), 235-263.
	
	
	
	\bibitem{CFPS}
	M. Cieliebak, P. Flocchini, G. Prencipe, N. Santoro, 
	Distributed computing by mobile robots: Gathering, SIAM J. Comput. 41 (2012), 829-879.
	
	
	\bibitem{CCGKM}
	A. Collins, J. Czyzowicz, L. Gasieniec, A. Kosowski, R. A. Martin,
	Synchronous rendezvous for location-aware agents. 
	Proc. 25th International Symposium on Distributed Computing (DISC 2011), 447-459.
	
	
	\bibitem{CKP}
	J. Czyzowicz, A. Kosowski, A. Pelc, How to meet when you forget: Log-space rendezvous in arbitrary graphs, Distributed Computing 25 (2012), 165-178. 
	
	
	
	
	
	
	%
	
	
	\bibitem{DFKP}
	A. Dessmark, P. Fraigniaud, D. Kowalski, A. Pelc.
	Deterministic rendezvous in graphs.
	Algorithmica 46 (2006), 69-96.
	
	\bibitem{DPV}
	Y. Dieudonn\'{e}, A. Pelc, V. Villain, How to meet asynchronously at polynomial cost, 
	SIAM Journal on Computing 44 (2015), 844-867. 
	
	
	
	
	%
	
	
	
	
	
	
	
	
	\bibitem{fpsw}
	P. Flocchini, G. Prencipe, N. Santoro, P. Widmayer,
	Gathering of asynchronous robots with limited visibility, Theoretical Computer Science 337 (2005), 147-168.
	
	\bibitem{FP}
	P. Fraigniaud, A. Pelc, Delays induce an exponential memory gap for rendezvous in trees, ACM Transactions on Algorithms 9 (2013), 17:1-17:24. 
	
	
	%
	%
	%
	%
	%
	%
	%
	%
	%
	%
	%
	%
	%
	
	
	
	%
	%
	
	\bibitem{KKSS}
	E. Kranakis, D. Krizanc, N. Santoro and C. Sawchuk, 
	Mobile agent rendezvous in a ring, 
	Proc. 23rd Int. Conference on Distributed Computing Systems
	(ICDCS 2003), 592-599.
	
	
	\bibitem{MP}
	A. Miller, A. Pelc, Fast deterministic rendezvous in labeled lines, Proc. 37th International Symposium on Distributed Computing (DISC 2023), 29:1 - 29:22. 
	
	%
	%
	%
	%
	%
	%
	
	%
	
	
	
	\bibitem{Pe2}
	A. Pelc, Deterministic rendezvous algorithms, in: Distributed Computing by Mobile Entities, P. Flocchini, G. Prencipe, N. Santoro, Eds., Springer 2019, LNCS 11340. 
	
	%
	\bibitem{PP2}
	D. Pattanayak, A. Pelc,
	Computing Functions by Teams of Deterministic Finite Automata. CoRR abs/2310.01151 (2023).
	
	\bibitem{Stach}
	G. Stachowiak, Asynchronous Deterministic Rendezvous on the Line
	Proc. 35th Conference on Current Trends in Theory and Practice of Computer Science (SOFSEM 2009), 497-508.	
	%
	
	\bibitem{TSZ07}
	A. Ta-Shma and U. Zwick.
	Deterministic rendezvous, treasure hunts and strongly universal exploration sequences.
	Proc. 18th ACM-SIAM Symposium on Discrete Algorithms (SODA 2007), 599-608.
	
	
	
	
	%
	%
	
\end{thebibliography}
\end{document}